\documentclass[letter,12pt]{amsart}
\usepackage{wrapfig}
\usepackage{amsfonts}
\usepackage{amssymb}
\usepackage{amsmath,tabu}
\usepackage{comment}
\usepackage{amsthm}
\usepackage{graphicx}
\usepackage{epstopdf}
\usepackage{pst-all,psfrag}
\usepackage[unicode]{hyperref}
\usepackage{indentfirst}
\usepackage{natbib}
\usepackage{verbatim}
\usepackage{eurosym}
\usepackage[utf8]{inputenc}
\usepackage[english]{babel}
\usepackage[margin=1in]{geometry}
\usepackage{chngcntr}
\usepackage{apptools}
\usepackage{tabularx}
\usepackage{setspace}
\usepackage{enumitem}
\usepackage{subcaption}

\newtheorem{theorem}{Theorem}
\newtheorem{proposition}[theorem]{Proposition}

\newtheorem{lemma}[theorem]{Lemma}

\newtheorem{assumption}{Assumption}

\newtheorem{corollary}[theorem]{Corollary}
\newtheorem{remark}[theorem]{Remark}

\setlist[enumerate,1]{label=\upshape{(\roman*)}, ref=(\roman*)}
\setlist[enumerate,2]{label=\upshape{(\alph*)}, ref=(\alph*)}
\setlist[enumerate,3]{label=\upshape{\roman*.}, ref=\roman*}

\newcommand{\eps}{\varepsilon}
\newcommand{\y}{\mathbf y}
\newcommand{\p}{\mathbf p}
\newcommand{\z}{\mathbf z}
\newcommand{\g}{\mathbf g}

\newcommand{\1}{I}
\newcommand{\I}{\mathcal I}
\newcommand{\F}{\mathcal F}
\newcommand{\T}{\mathsf T}

\begin{document}

\author{Anna Bykhovskaya}
\address[Anna Bykhovskaya]{Duke University}
\email{anna.bykhovskaya@duke.edu}

\author{Nour Meddahi}
\address[Nour Meddahi]{Toulouse School of Economics}
\email{nour.meddahi@tse-fr.eu}

\title{Generalized Autoregressive Multivariate Models:\\ From Binary to Poisson}

\begin{abstract}
This paper presents a framework for binary autoregressive time series in which each observation is a Bernoulli variable whose success probability evolves with past outcomes and probabilities, in the spirit of GARCH-type dynamics, accommodating nonlinearities, network interactions, and cross-sectional dependence in the multivariate case. Existence and uniqueness of a stationary solution is established via a coupling argument tailored to the discontinuities inherent in binary data. A key theoretical result, further supported by our empirical illustration on S$\&$P $100$ data, shows that, under a rare-events scaling, aggregates of such binary processes converge to a Poisson autoregression, providing a micro-foundation for this widely used count model. Maximum likelihood estimation is proposed and illustrated empirically.

\medskip{}
\end{abstract}

\thanks{}

\date{\today}
\maketitle

\section{Introduction}

Binary time series taking values in $\{0,1\}$ are central to modeling dynamic economic and financial phenomena in which outcomes represent discrete events or decisions rather than continuous quantities. In economics, they arise naturally in dynamic binary choice models, where agents repeatedly decide whether to take actions such as purchasing, investing, or entering a market \citep{rust1987optimal, heckman1981heterogeneity}. In finance, binary indicators capture rare but economically important events, such as extreme return days and bear markets \citep{nyberg2013predicting}, as well as corporate defaults and credit downgrades \citep{duffie2009frailty}. Macroeconomic forecasting provides another key application, where the variable of interest is binary, such as the onset of a recession \citep{kauppi2008predicting}.

Beyond individual processes, multivariate binary time series offer a natural framework for modeling interdependent or networked decisions, where outcomes for one unit depend on others. Such models are useful for analyzing financial contagion, correlated defaults, and evolving networks of economic linkages. Allowing the network structure to vary over time further enriches the analysis by capturing changing patterns of influence, information diffusion, or systemic risk propagation (cf., dynamic network formation of \citealt{graham2016homophily}). Aggregating binary outcomes across a large panel produces a count of successes at each date. We show that, under a rare-events scaling, these cross-sectional interactions provide a natural microfoundation for Poisson autoregression, a leading model for count time series.

We develop a general framework for multivariate binary time series in which each observation follows a Bernoulli distribution with a time-varying success probability that depends on both past outcomes and past success probabilities, in the spirit of the GARCH-type dynamics of \citet{bollerslev1986generalized}. Unlike the models of \citet{kedem2005regression} and \citet{de2011dynamic}, our formulation allows current success probabilities to depend on latent past success probabilities, rather than only on observed outcomes. This additional layer accommodates richer persistence and feedback dynamics. Under general and potentially nonlinear dynamics for the success probabilities, we establish the existence of a stationary solution to our model. Under a contraction condition on the probability-update function, we further establish its uniqueness.

Our model fits into several important classes. When specialized to a linear form, our model corresponds to the generalized autoregressive score (GAS) framework of \citet{creal2013generalized}. However, most of the existing results for GAS models \citep[e.g.,][]{blasques2014stationarity, blasques2022maximum} do not apply in our setting. The reason is that Bernoulli variables, expressible as indicators that a uniform random shock falls below a success probability, generate discontinuities that violate the smoothness and Lipschitz assumptions crucial for the GAS theory. Our model is also a member of the class of positive-valued time series with exponential-family conditional distributions studied, in the univariate case, by \citet{davis2016theory,aknouche2021count} and, in the multivariate case with one lag of dependent variables, by \citet{lee2023modeling}. Our analysis departs from those contributions by focusing specifically on cross-sectional interactions and their effect on the aggregate number of successes, a question those papers do not address.

The stochastic properties of related binary models have been studied by \citet{moysiadis2014binary} and \citet{fokianos2017binary}, and some of our structural results, particularly the existence and uniqueness of a stationary solution, are analogous in spirit to results obtained there, reflecting the robustness of the underlying probabilistic argument. The present paper departs from that work in two important respects. First, we consider an $N$-dimensional panel with cross-sectional interactions. Second, and more substantially, we target the rare-events asymptotic regime, which requires success probabilities to be close to zero and motivates working directly with $[0,1]^{N}$ rather than reparametrizing the probabilities via an inverse distribution function. The cited papers apply a smooth monotone transformation, such as the logistic or normal CDF, that maps probabilities to an unbounded domain. Under these transformations, small probabilities are mapped to values near $-\infty$, making analysis in this region problematic. Moreover, when combined additively, these highly negative transformed values decrease further, behaving more like products of probabilities than sums, which complicates interpretation.

Our central result provides a micro-level foundation for the integer-valued GARCH (INGARCH) model \citep{rydberg2000modelling,streett2000some,ferland2006integer}, also known as Poisson autoregression \citep{fokianos2009poisson}, one of the most widely used models for count time series. Under a rare-events scaling in which individual success probabilities vanish as the cross-sectional dimension $N$ grows, we show that the aggregate number of successes converges, in finite-dimensional distributions, to a Poisson autoregression.

The practical motivation is clear: many economically important phenomena, such as defaults in large credit portfolios, the exercise of binary options across many positions, or episodes of heightened systemic risk, are individually rare, yet their cross-sectional aggregate incidence is of primary interest. This contrasts with the original motivation of \citet{rydberg2000modelling}, which arose from discretizing time in a Poisson process.

Our setting yields a richer family of Poisson autoregressions than the parametric specifications considered in \cite{fokianos2009poisson}. In particular, our recursion for the Poisson intensity admits genuinely nonlinear functions of both the aggregate count and the limiting intensity. We also show that the result extends to network interaction structures: when the network is doubly stochastic and sufficiently dense, heterogeneous interaction weights wash out asymptotically, and the same Poisson autoregressive limit obtains. This illustrates a broader phenomenon in which network topology becomes irrelevant at the aggregate level once connectivity is sufficiently rich.

Finally, we propose maximum likelihood estimation (MLE) and establish consistency and asymptotic normality under a boundedness condition on the success probabilities, which rules out the rare-events regime. In that regime the Poisson MLE applied to the aggregate count is the natural substitute, and under structural assumptions such as homogeneity, the aggregation theory supplies an explicit mapping from the Poisson parameters back to those of the underlying binary system.

We illustrate the framework using daily returns on S$\&$P 100 constituents over 2012--2025, defining tail events as idiosyncratic returns falling below the $5$th percentile. The binary and Poisson parameter estimates are close, consistent with the aggregation theory, and the filtered Poisson intensity tracks the sum of individual success probabilities closely in sample. In out-of-sample forecasting, the heterogeneous binary model outperforms the Poisson-calibrated specification, indicating that the granularity of the binary likelihood remains empirically meaningful.

\smallskip

The remainder of the paper is organized as follows. Section \ref{sec_setting} introduces the model, provides motivating examples, and establishes the existence and uniqueness of a stationary solution. Section \ref{sec_aggregation} analyzes the model's aggregate behavior, focusing on the total number of successes within a given period. Section \ref{sec_estimation} derives the consistency and asymptotic normality of the maximum likelihood estimator. Section \ref{sec_covariates} extends the model by allowing for additional covariates. Section \ref{sec_illustr} applies the proposed framework to S$\&$P 100 data. Finally, Section \ref{sec_conclusion} concludes. All proofs appear in the appendices.

\section{Setting}\label{sec_setting}

Consider an $N$-dimensional binary vector $\y_t=(y_{1,t},\ldots,y_{N,t})^{\T}$, where the superscript $\T$ denotes transposition. Here time is discrete and we deal with two cases: either two-sided $t\in \mathbb Z$ or one-sided $t\in\{t_0,t_0+1,t_0+2,\ldots\}$ for some $t_0\in\mathbb Z$.
 Let $\F_t$ be a filtration of $\sigma$-fields, so that $\y_t$ is $\F_t$-measurable.
The components of $\y_t$ are assumed to be independent conditional on $\F_{t-1}$.\footnote{Contemporaneous conditional dependence can be accommodated without affecting the stationarity results. The aggregation results, however, require an appropriate form of weak cross-sectional dependence. Such dependence can be introduced by replacing the i.i.d.~$U[0,1]$ innovations below with innovations that are conditionally marginally uniform on $[0,1]$ but cross-sectionally dependent.} Given $\F_{t-1}$, each component $y_{i,t}$, $i=1,\ldots,N$, is assumed to be a Bernoulli random variable with success probability $p_{i,t}$, measurable with respect to $\F_{t-1}$,
\begin{equation}\label{eq_dgp_y}
y_{i,t}\mid \F_{t-1}\thicksim B(p_{i,t}) \text{ or equivalently }y_{i,t}= \1(u_{i,t}\leq p_{i,t}),\,u_{i,t}\mid \F_{t-1}\thicksim i.i.d.~U[0,1].
\end{equation}

The vector of success probabilities $\p_t=(p_{1,t},\ldots,p_{N,t})^{\T}$ may depend on its own past values $\p_{\tau},\,\tau=t-1,\ldots,t-s$, as well as on past realizations $\y_{\tau},\,\tau=t-1,\ldots,t-q$. Formally,
\begin{equation}\label{eq_dgp_p}
  \p_{t}=\g\left(\{\p_{\tau}\}_{\tau=t-1,\ldots,t-s},\,\{\y_{\tau}\}_{\tau=t-1,\ldots,t-q}\right),
\end{equation}
where $s$ and $q$ are finite integers controlling the memory of the process, and $\g(\cdot)$ is a continuous mapping from $[0,1]^{Ns}\times\{0,1\}^{Nq}$ to $[0,1]^N$.
When the time is one-sided, $t\in\{t_0,\ldots\}$, the evolution \eqref{eq_dgp_p} should be supplemented with initial values of $\p_{t}$ and the equation is only valid once all the indices in \eqref{eq_dgp_p} are at least $t_0$.

The general form of $\g(\cdot)$ allows for nonlinearities and interactions, whereby $y_{i,t}$ depends not only on its own past but also on the past values and success probabilities of all $N$ variables. The dependence of $\p_t$ on both $\y_{t-1}$ and $\p_{t-1}$ mirrors the GARCH structure of \citet{bollerslev1986generalized}. We, therefore, refer to the dynamics in \eqref{eq_dgp_y}--\eqref{eq_dgp_p} as the \textit{Generalized Autoregressive Binary Process} (GAB).

The Bernoulli distribution belongs to the one-parameter exponential family and, thus, the GAB model fits in the class of models considered in \citet{davis2016theory,aknouche2021count,lee2023modeling}.

\subsection{Examples}\label{sec_examples}
Let us introduce some useful examples of the function $\g(\cdot)$. We begin with a simpler setting where $N=1$, so that $\p_t=p_{1,t}\equiv p_t$ and $\y_t=y_{1,t}\equiv y_t$.
\begin{itemize}
  \item[] {\bf Example 1:} Linear GAB$(1,1)$
  \begin{equation}\label{eq_BGARCH11}
  p_t=\omega+\alpha y_{t-1}+\beta p_{t-1},\quad \omega,\alpha,\beta\geq0,\quad \omega+\alpha+\beta\leq1.
  \end{equation}
  The conditions on the parameters ensure that $p_t$ stays in $[0,1]$. This framework resembles a classical GARCH model, where instead of conditional variance we specify conditional success probability as a linear function of its own past and a past variable realization. Moreover, the model \eqref{eq_BGARCH11} is a special case of generalized autoregressive score (GAS) models of \citet{creal2013generalized}, where, in the notations of \citep[]{creal2013generalized}, $s_t=y_t-p_t,\,S_t=p_t(1-p_t)$.
  \item[]  {\bf Example 2:}  Linear GAB$(s,q)$
  \begin{equation}\label{eq_BGARCHsq}
  p_t=\omega+\sum\limits_{j=1}^{q}\alpha_j y_{t-j}
  +\sum\limits_{j=1}^{s}\beta_j p_{t-j},\quad
  \omega,\alpha_j,\beta_j\geq0,\quad \omega+\sum\limits_{j=1}^{q}\alpha_j+\sum\limits_{j=1}^{s}\beta_j\leq1.
  \end{equation}
  This framework extends the previous one by allowing for multiple lags of both probabilities and outcomes.
  \item[]  {\bf Example 3:}  Logit GAB$(1,1)$. Consider a function from $[0,1]\times\{0,1\}$ to $[0,1]$, defined for any parameter values $\omega,\alpha,\beta$,
  $$\g(p,y)=\left(1+\exp(-\omega-\alpha y)\left(\frac{1-p}{p}\right)^{\beta}\right)^{-1}.$$
  The above function does not require additional parameter restrictions to ensure that probabilities remain in $[0,1]$. This function is continuous and maps $(0,0)$ to $0$ and $(1,1)$ to $1$ when $\beta>0$, while for $\beta<0$ it maps $(0,0)$ to $1$ and $(1,1)$ to $0$. In contrast to the logistic transformation of \citet{moysiadis2014binary,fokianos2017binary}, the above form still allows for small probabilities. In terms of our original variables, the evolution can be written as
  \begin{equation}\label{eq_logitGARCH}
    p_t=\left(1+\exp(-\omega-\alpha y_{t-1})\left(\frac{1-p_{t-1}}{p_{t-1}}\right)^{\beta}\right)^{-1}.
  \end{equation}
  Equivalently, one has
  $$\ln\left(\frac{p_t}{1-p_t}\right)=\omega+\alpha y_{t-1}+\beta \ln\left(\frac{p_{t-1}}{1-p_{t-1}}\right)$$
  or, for $x_t=\ln\left(\frac{p_t}{1-p_t}\right)\in\mathbb{R}\cup\{\pm\infty\}$, so that $p_t=(1+\exp(-x_t))^{-1}$,
  \begin{equation}\label{eq_logitGARCHx}
  x_t=\omega+\alpha y_{t-1}+\beta x_{t-1}.
  \end{equation}
  Eq.~\eqref{eq_logitGARCHx} corresponds to \citet{moysiadis2014binary,fokianos2017binary}, except that they restrict $x_t\in\mathbb{R}$, ruling out $p_t=0$.
  \item[]  {\bf Example 4:}  Nonlinear GAB$(1,1)$
  $$p_t=\omega+\alpha y_{t-1} + f(p_{t-1}),$$
  where $f(\cdot)$ is a continuous function such that $\forall p\in[0,1]$,
  $$-\omega-\min\{\alpha,0\}\leq f(p)\leq 1-\omega-\max\{\alpha,0\}.$$
  This framework allows for nonlinear dependence between $p_t$ and $p_{t-1}$. Notice that, since $y_t$ is binary, we do not need to apply any nonlinear function to it.

  Figures \ref{fig_GAB_simulation_small_nonlin} and \ref{fig_GAB_simulation_small_nonlin2} illustrate nonlinear GAB$(1,1)$ models with $\omega = 0.05$ and $\alpha = 0.4$, using $f(p) = 2p(p - 0.5)^2$ and $f(p) = (p + 0.7)(p - 5/6)^2$, respectively. The key difference lies in the magnitude of the nonlinear component. In Figure \ref{fig_GAB_simulation_small_nonlin}, the nonlinear term remains small, especially when $p \approx 0$, resulting in oscillatory dynamics with prolonged low-probability periods punctuated by short-lived spikes following realizations of $y_t = 1$. Overall, the path of $p_t$ stays close to the Markovian component $\omega + \alpha y_{t-1}$. In contrast, Figure \ref{fig_GAB_simulation_small_nonlin2} exhibits a much stronger nonlinear effect, leading to substantially different dynamics and pronounced deviations from the Markovian benchmark.

\end{itemize}

\begin{figure}[t]
\centering
\begin{subfigure}[t]{0.48\linewidth}
    \centering
    \includegraphics[width=\linewidth]{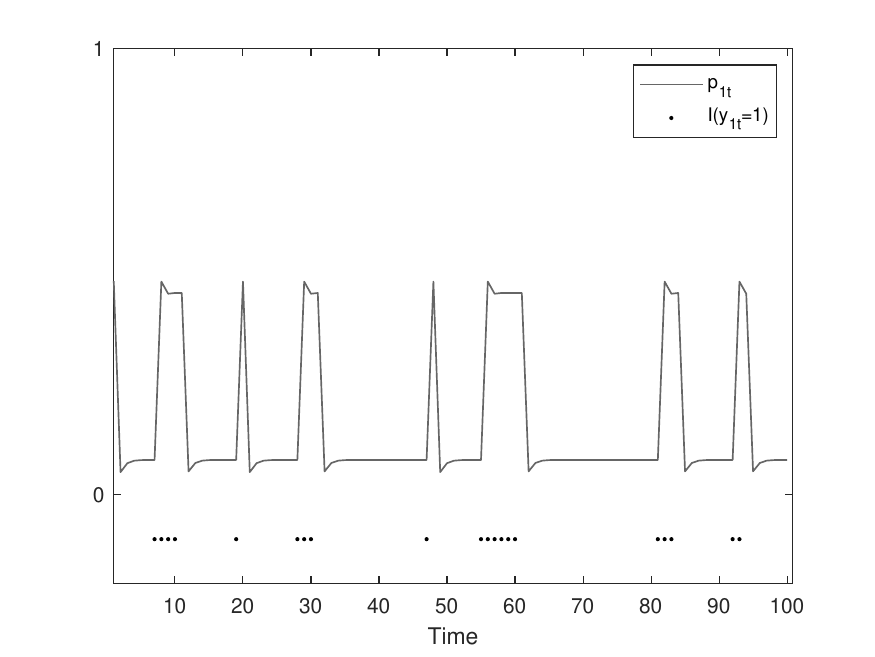}
    \caption{Nonlinear single-series model: $p_{1,0}=1/12$, ${p_{1,t}=0.05+0.4y_{1,t-1}+2p_{1,t-1}(p_{1,t-1}-0.5)^2}$.}
    \label{fig_GAB_simulation_small_nonlin}
\end{subfigure}
\hfill
\begin{subfigure}[t]{0.48\linewidth}
    \centering
    \includegraphics[width=\linewidth]{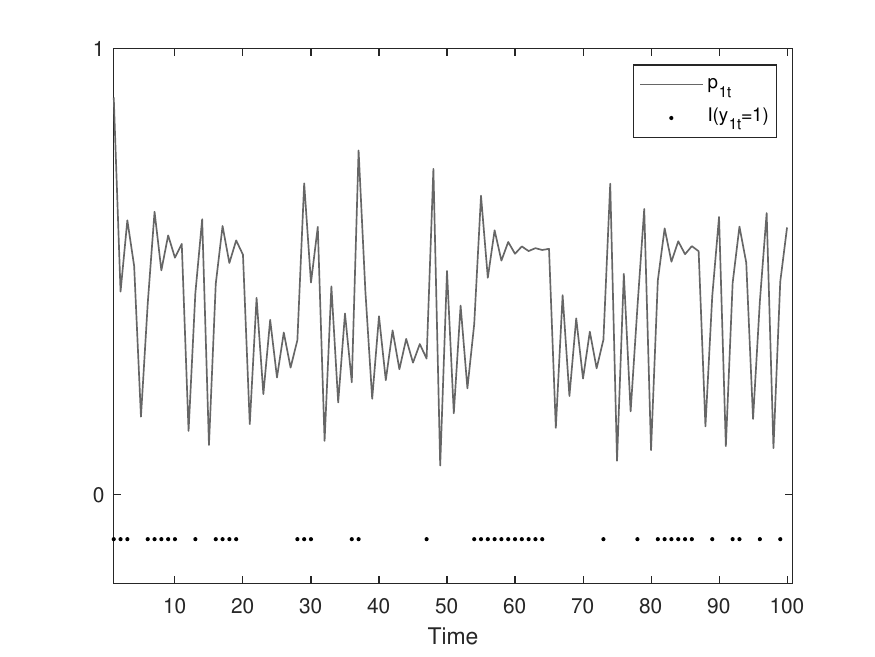}
    \caption{Nonlinear single-series model: $p_{1,0}=1/12$, ${p_{1,t}=0.05+0.4y_{1,t-1}+(p_{1,t-1}+0.7)(p_{1,t-1}-5/6)^2}$.}
    \label{fig_GAB_simulation_small_nonlin2}
\end{subfigure}

\begin{subfigure}[t]{0.48\linewidth}
    \centering
    \includegraphics[width=\linewidth]{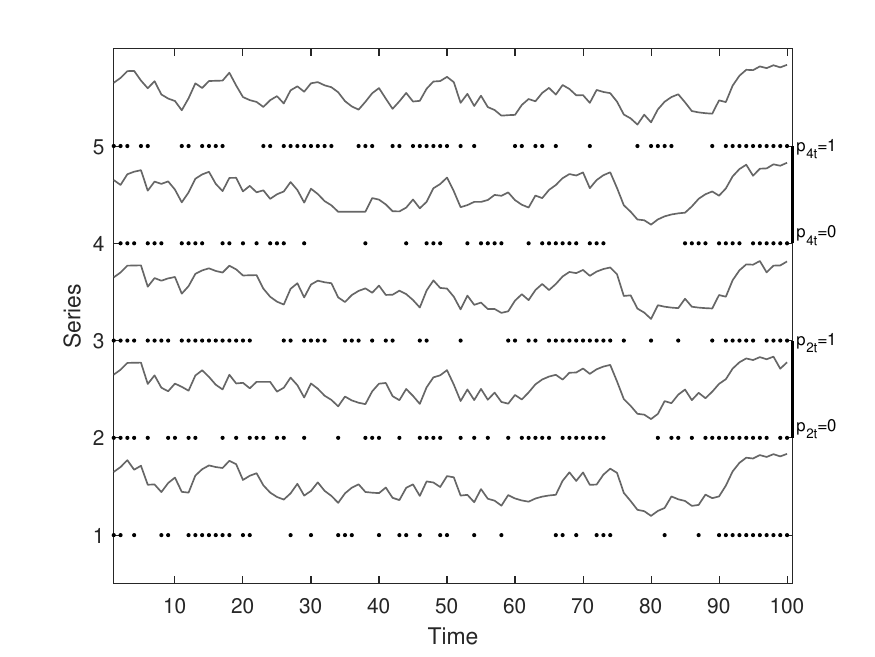}
    \caption{Interactive model \eqref{eq_multiple_k_diff} with $N=5$, ${p_{i,0}=0.5}$, and $\omega_i=0.05$, $\alpha_i=0.1$, $\beta_i=0.6$, $\gamma_i=0.2$ for all $i$.}
    \label{fig_GAB_simulation_small}
\end{subfigure}
\hfill
\begin{subfigure}[t]{0.48\linewidth}
    \centering
    \includegraphics[width=\linewidth]{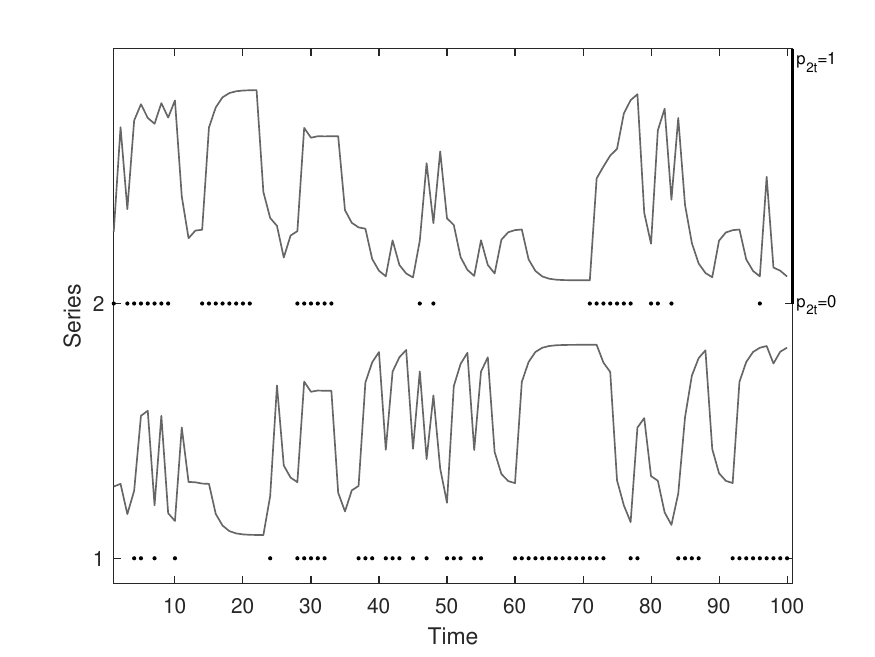}
    \caption{Two-series model with negative interaction: $N=2$, $p_{1,0}=p_{2,0}=3/7$, $p_{i,t}=0.3+0.4y_{i,t-1}+0.2p_{i,t-1}-0.3p_{j,t-1}$, $i\neq j$.}
    \label{fig_GAB_simulation_small_neg}
\end{subfigure}

\caption{Raster plots with probability trajectories. Black dots indicate success periods ($y_{i,t}=1$). Continuous lines represent the corresponding success probabilities $p_{i,t}$.}
\label{fig_GAB_simulation_combined}
\end{figure}

All of the above examples can be straightforwardly generalized to $N>1$, where each time series $i$ has its own $p_{i,t}$, whose evolution involves only quantities indexed by $i$ and no quantities indexed by $j\neq i$. More interesting cases arise when the dynamics of $p_{i,t}$ depend on the outcomes or probabilities of other series. A few examples are provided below.
\begin{itemize}
  \item[]  {\bf Example 5:}  Exchangeable GAB$(1,1)$. All $N$ components share a common probability of success, driven by the cross-sectional average of past outcomes,
  \begin{equation}\label{eq_multiple_k}
  p_{i,t}\equiv p_t=\omega+\gamma \frac{1}{N}\sum_{j=1}^N y_{j,t-1}+\beta p_{t-1},\quad \omega,\beta,\gamma\geq0,
  \quad \omega+\beta+\gamma\leq1.
  \end{equation}
  Under this exchangeable specification, the sum $\sum_{j=1}^N y_{j,t}$, conditional on $\I_{t-1}$, follows a binomial distribution with parameters $N$ and $p_t$.
  \item[]  {\bf Example 6:} Interactive GAB$(1,1)$. Each component has its own probability of success, driven by both its own past outcome and the cross-sectional average,
  \begin{equation}\label{eq_multiple_k_diff}
    p_{i,t}=\omega_i+\alpha_i y_{i,t-1}+ \gamma_i \frac{1}{N}\sum_{j=1}^N y_{j,t-1}+\beta_i p_{i,t-1},\quad \omega_i,\alpha_i,\gamma_i,\beta_i\geq0,\quad \omega_i+\alpha_i+\gamma_i+\beta_i\leq1.
  \end{equation}
  When $\gamma_i=0$ for all $i$, the model \eqref{eq_multiple_k_diff} reduces to $N$ trajectories of GAB$(1,1)$ models \eqref{eq_BGARCH11} that are independent conditional on $\I_{t-1}$. When, for all $i$, $\omega_i\equiv\omega,\,\alpha_i\equiv0,\,\beta_i\equiv\beta,\,\gamma_i\equiv\gamma$, the model \eqref{eq_multiple_k_diff} reduces to a single probability model (\ref{eq_multiple_k}).
  \item[]  {\bf Example 7:}  Network GAB$(1,1)$. Each component has its own probability of success, driven by a network-weighted average of past outcomes,
  \begin{equation}\label{eq_multiple_k_diff_W}
    p_{i,t}=\omega_i+\alpha_i y_{i,t-1}+ \gamma_i \sum_{j=1}^N W_{ij}y_{j,t-1}+\beta_i p_{i,t-1},\quad \omega_i,\alpha_i,\gamma_i,\beta_i\geq0,\quad \omega_i+\alpha_i+\gamma_i+\beta_i\leq1,
  \end{equation}
  where $W$ is an $N\times N$ row-normalized adjacency matrix governing the relationships among individuals. Eq.~\eqref{eq_multiple_k_diff} corresponds to the special case $W_{ij}=1/N$ for all $i,j$, i.e., the complete graph. 
\end{itemize}
All of the above multivariate examples can be readily extended to incorporate multiple lags of outcomes and probabilities, analogous to the transition from GAB($1,1$) to GAB($s,q$). They can also be generalized to allow for nonlinearities that aggregate $p_{i,t}$ and/or $y_{i,t}$ across $i$, for instance, in the spirit of the peer-effect interaction terms of \citet{bykhovskaya2022time}, where a nonlinear function of neighbors' past outcomes enters the dynamics.

Figures \ref{fig_GAB_simulation_small} and \ref{fig_GAB_simulation_small_neg} visualize multivariate GAB dynamics for $y_{i,t}$ and $p_{i,t}$: dots at row $i$ indicate realizations of successes for series $i$, while the continuous trajectory above each row represents the corresponding success probability $p_{i,t}$. Figure \ref{fig_GAB_simulation_small} corresponds to the interactive model \eqref{eq_multiple_k_diff} with $N=5$ and strictly positive coefficients $\gamma_i$. The figure reveals pronounced co-movement in the success probabilities induced by the common component $\frac{1}{N}\sum_{j=1}^N y_{j,t-1}$, while the series-specific term $\alpha_i y_{i,t-1}$ introduces heterogeneity across series.

Figure \ref{fig_GAB_simulation_small_neg} illustrates, using two series, how the opposite effect can arise, with the success probabilities moving in opposite directions. In particular, when the success probability $p_{j,t}$ increases, the negative term $-0.3 p_{j,t}$ becomes more pronounced,\footnote{The coefficients are chosen to ensure that $p_{i,t} \in [0,1]$.} which in turn reduces the success probability $p_{i,t+1}$ for $i \neq j$.

\subsection{Stationarity}
We next analyze the stationarity properties of the GAB model \eqref{eq_dgp_y}–\eqref{eq_dgp_p} in the asymptotic regime $T \to \infty$, with the number of time series $N$ held fixed.
\subsubsection{Existence}
\begin{theorem}\label{th_stationary_sol}
For any continuous  $\g(\cdot)$, there exists a probability space supporting $\F_t$ and strictly stationary sequence $(\y_t,\p_t)_{t\in\mathbb Z}$, such that \eqref{eq_dgp_y}--\eqref{eq_dgp_p} holds.
\end{theorem}

Notice that no additional conditions are imposed, implying that all GAB models admit a stationary solution regardless of parameter values. The key features ensuring this property are the bounded support of all variables and the continuity of the probability function $\g$. Consequently, all examples from the preceding subsection, including the nonlinear Example 3 and the interactive Examples 6 and 7, possess stationary solutions. In the case of the logit GAB model (Example 3), note in particular that it is not necessary to impose the classical condition $|\beta|<1$.

Since all variables are bounded by construction (outcomes in $\{0,1\}$ and probabilities in $[0,1]$), the process has finite moments of all orders. Hence, strict stationarity implies weak stationarity, leading to the following corollary.
\begin{corollary}
The multivariate binary time series model \eqref{eq_dgp_y}--\eqref{eq_dgp_p}, $t\in\mathbb{Z}$, has a weakly stationary solution.
\end{corollary}

The law of iterated expectations and \eqref{eq_dgp_y} imply that $\mathbb{E}y_{i,t}=\mathbb{E}p_{i,t}$ and $\mathbb{E}_{t-1}y_{i,t}=p_{i,t},$ where the subscript $t-1$ denotes conditioning on $\F_{t-1}$. Assuming $\sum\limits_{j=1}^{q}\alpha_j+\sum\limits_{j=1}^{s}\beta_j<1$, the linear GAB($s,q$) yields
\begin{equation}\label{eq_BGARCH_mean}
  \mu:=\mathbb{E}y_{i,t}=\mathbb{E}p_{i,t}=\frac{\omega}{1-\sum\limits_{j=1}^{q}\alpha_j-\sum\limits_{j=1}^{s}\beta_j},
\end{equation}
and, assuming $\gamma+\beta<1$, its $N$-dimensional variation \eqref{eq_multiple_k} yields
$$
\mu:=\mathbb{E}y_{i,t}=\mathbb{E}p_{i,t}=\frac{\omega}{1-\gamma-\beta},
$$
which matches the unconditional variance formula for classical GARCH($s,q$) and GARCH($1,1$).

In the interactive model \eqref{eq_multiple_k_diff}, assuming $\alpha_i+\beta_i<1$ for all $i$ and $\frac1{N}\sum_{i=1}^N\frac{\gamma_i}{1-\alpha_i-\beta_i}<1$,
$$\mu_i:=\mathbb{E}y_{i,t}=\mathbb{E}p_{i,t}=\omega_i+ \alpha_i \mu_{i}+ \gamma_i \frac{1}{N}\sum_{j=1}^N \mu_{j}+\beta_i \mu_{i}
=\frac{\omega_i}{1-\alpha_i-\beta_i}+\frac{\gamma_i}{N(1-\alpha_i-\beta_i)}\sum_{j=1}^N \mu_{j},$$
so that
\begin{equation}\label{eq_multiple_k_diff_mean_sum}
\sum_{i=1}^N\mu_i=\sum_{i=1}^N \frac{\omega_i}{1-\alpha_i-\beta_i}
+ \frac1{N}\sum_{j=1}^N \mu_{j}\sum_{i=1}^N\frac{\gamma_i}{1-\alpha_i-\beta_i}
=\frac{\sum_{i=1}^N \frac{\omega_i}{1-\alpha_i-\beta_i}}{1-\frac1{N}\sum_{i=1}^N\frac{\gamma_i}{1-\alpha_i-\beta_i}}
\end{equation}
and
\begin{equation}\label{eq_multiple_k_diff_mean}
\mu_i=\frac{\omega_i}{1-\alpha_i-\beta_i}+\frac{\gamma_i}{N(1-\alpha_i-\beta_i)}\frac{\sum_{j=1}^N \frac{\omega_j}{1-\alpha_j-\beta_j}}{1-\frac1{N}\sum_{j=1}^N\frac{\gamma_j}{1-\alpha_j-\beta_j}}.
\end{equation}

In the network GAB$(1,1)$ model \eqref{eq_multiple_k_diff_W}, a similar argument yields, in matrix form with $A:=\mathrm{diag}(\alpha_i+\beta_i)$, $\Gamma:=\mathrm{diag}(\gamma_i)$, $\boldsymbol\omega=(\omega_1,\ldots,\omega_N)^{\T}$, and $\boldsymbol\mu=(\mu_1,\ldots,\mu_N)^{\T}$,
\begin{equation}\label{eq_network_mean}
(I - A - \Gamma W)\boldsymbol\mu = \boldsymbol\omega.
\end{equation}
Since $W$ is row-stochastic, the row sums of $A+\Gamma W$ equal $\alpha_i+\beta_i+\gamma_i\leq 1-\omega_i$, so for $\omega_i>0$ the matrix $I-A-\Gamma W$ is strictly row-diagonally dominant and hence invertible, giving $\boldsymbol\mu=(I-A-\Gamma W)^{-1}\boldsymbol\omega$. For general $W$ this solution does not simplify further. The interactive model \eqref{eq_multiple_k_diff} corresponds to $W_{ij}=1/N$, the complete graph, in which case the uniform structure allows $\sum_i\mu_i$ to be solved explicitly, yielding the closed-form expressions \eqref{eq_multiple_k_diff_mean_sum}--\eqref{eq_multiple_k_diff_mean}.

The linear GAB($s,q$) model with $\sum\limits_{j=1}^{q}\alpha_j+\sum\limits_{j=1}^{s}\beta_j=1$, and hence $\omega=0$, admits multiple stationary solutions. In particular, both $y_{it}=p_{it}\equiv0$ and $y_{it}=p_{it}\equiv1$ constitute trivial stationary paths. This is in contrast to the IGARCH case, where the unit-root condition destroys weak stationarity; here, bounded support guarantees that stationarity is preserved even on the boundary, though at the cost of uniqueness.

\subsubsection{Uniqueness}\label{sec_unique}

\begin{assumption}\label{ass_contraction1}
Suppose that there exist non-negative coefficients $\{\alpha_{ij}^{\tau}\}$, $\{\beta_{ij}^{\tau}\}$, with $\beta_{ij}^{\tau}=0$ for $\tau>s$ and
$\alpha_{ij}^{\tau}=0$ for $\tau>q$, such that the function $\g(\cdot)=(g_1(\cdot),\ldots,g_N(\cdot))^{\T}$ satisfies the Lipschitz bound
\begin{equation}\begin{split}\label{eq_g_contraction1}
 &|g_i\left(\{\p_{\tau}\}_{\tau=t-1,\ldots,t-s},\,\{\y_{\tau}\}_{\tau=t-1,\ldots,t-q}\right)-
 g_i\left(\{\p_{\tau}'\}_{\tau=t-1,\ldots,t-s},\,\{\y_{\tau}'\}_{\tau=t-1,\ldots,t-q}\right)|\\
 &\quad\leq
 \sum\limits_{j=1}^{N}\sum\limits_{\tau=1}^{s} |p_{j,t-\tau}-p_{j,t-\tau}'|\beta_{ij}^{\tau}
 +\sum\limits_{j=1}^{N}\sum\limits_{\tau=1}^{q} |y_{j,t-\tau}-y_{j,t-\tau}'|\alpha_{ij}^{\tau}\,\quad\forall i.
\end{split}\end{equation}
Let $\Phi$ be the associated $N\max(s,q)\times N\max(s,q)$ companion matrix, where the first $N$ rows of $\Phi$ are given by
$$\left(\alpha_{i1}^{1}+\beta_{i1}^{1},\ldots,\alpha_{iN}^{1}+\beta_{iN}^{1},\ldots \alpha_{i1}^{\max(s,q)}+\beta_{i1}^{\max(s,q)},\ldots,\alpha_{iN}^{\max(s,q)}+\beta_{iN}^{\max(s,q)}\right),
\quad i=1,\ldots,N,$$
the bottom-left $N(\max(s,q)-1)\times N(\max(s,q)-1)$ block of $\Phi$ is an identity matrix, and the bottom-right $N(\max(s,q)-1)\times N$ block is zero. Suppose that the spectral radius $\rho(\Phi)<1$.
\end{assumption}
Assumption \ref{ass_contraction1} requires that the Lipschitz coefficients of $\g(\cdot)$, when arranged into the companion matrix $\Phi$, yield a stable system in the sense that the spectral radius $\rho(\Phi)<1$. This is the direct analogue of the stationarity condition for a vector autoregression of order $\max(s,q)$. A simple
sufficient condition (and equivalent for $N=1$) is the row-sum bound $\sum_{j=1}^{N}\sum_{\tau=1}^{s}\beta_{ij}^{\tau}
+\sum_{j=1}^{N}\sum_{\tau=1}^{q}\alpha_{ij}^{\tau}\leq K<1$ for all $i$, which implies $\rho(\Phi)<1$. For the network GAB$(1,1)$ model \eqref{eq_multiple_k_diff_W}, the companion matrix reduces to the single $N\times N$ block $A+\Gamma W$, so Assumption \ref{ass_contraction1} specializes to $\rho(A+\Gamma W)<1$, which also implies invertibility of $I-A-\Gamma W$ and hence the existence of the closed-form mean \eqref{eq_network_mean}.

Alternatively, the contraction assumption can be stated in an aggregate form:
\begin{assumption}\label{ass_contraction2}
There exists $0\leq K<\frac{1}{s+q}$ such that
\begin{equation}\begin{split}\label{eq_g_contraction2}
 &\sum\limits_{i=1}^{N}|g_i\left(\{\p_{\tau}\}_{\tau=t-1,\ldots,t-s},\,\{\y_{\tau}\}_{\tau=t-1,\ldots,t-q}\right)-
 g_i\left(\{\p_{\tau}'\}_{\tau=t-1,\ldots,t-s},\,\{\y_{\tau}'\}_{\tau=t-1,\ldots,t-q}\right)|\\
 &\leq
 K\left(\sum\limits_{i=1}^{N}\sum\limits_{\tau=1}^{s} |p_{i,t-\tau}-p_{i,t-\tau}'|
 +\sum\limits_{i=1}^{N}\sum\limits_{\tau=1}^{q} |y_{i,t-\tau}-y_{i,t-\tau}'|\right).
\end{split}\end{equation}
\end{assumption}

Neither Assumption \ref{ass_contraction1} nor Assumption \ref{ass_contraction2} implies the other. 
For example, when $N=2$ and $s=q=1$, let $g_1(\p,\y)=\tfrac{2}{3}p_1$ and $g_2(\p,\y)=0$, then Assumption \ref{ass_contraction1} holds, but Assumption \ref{ass_contraction2} fails since $2/3>1/2$. Conversely, let $g_i(\p,\y)=f_i(p_1)+f_i(p_2)+f_i(y_1)+f_i(y_2)$, where $f_1(x)=\min\{0.2,0.4x\}$, $f_2(x)=\max\{0,0.4(x-0.5)\}$. Then $g_i\in[0,0.8]$ and each $f_i$ is Lipschitz with constant $0.4$. Moreover, the intervals on which $f_1$ and $f_2$ are nonconstant do not overlap, so bound \eqref{eq_g_contraction2} holds with $K = 0.4$ (and no smaller constant is possible). However, under Assumption \ref{ass_contraction1}, we must have $\alpha_{ij} = \beta_{ij} = 0.4$, implying $\Phi = 0.8 (1,1)^{\T}(1,1)$ and hence $\rho(\Phi) = 1.6$.

\begin{theorem}\label{th_stationary_unique}
Suppose that  $\g(\cdot)$ satisfies either Assumption \ref{ass_contraction1} or \ref{ass_contraction2}. The distribution of a strictly stationary sequence $(\y_t,\p_t)_{t\in\mathbb Z}$ solving \eqref{eq_dgp_y}--\eqref{eq_dgp_p} is unique. If the time is one-sided, $t\in\{t_0,\ldots\}$, then the time-series forgets its initial condition geometrically fast and converges in distribution to the unique stationary solution as $t\to\infty$.
\end{theorem}

More precisely, by forgetting the initial condition we mean that for any two realizations $(\y_t,\p_t)$ and $(\y'_t,\p'_t)$, $t\ge t_0$, of the process driven by the same innovation sequence $u_{i,t}$, but initialized from different values, we have
\begin{equation}
\label{eq_geometeric_mixing}
\mathbb E\sum_{i=1}^{N}|y_{i,t}-y'_{i,t}|=\mathbb E\sum_{i=1}^{N}|p_{i,t}-p'_{i,t}| \leq C \varrho^t,\quad t\ge t_0,
\end{equation}
for some constants $C<\infty$ and $0\leq\varrho<1$. 

While we could not locate Theorem \ref{th_stationary_unique} in the literature, related results have appeared in \citet{moysiadis2014binary,davis2016theory,fokianos2017binary,aknouche2021count,lee2023modeling}.
The nontrivial aspect in the binary setting arises from the discontinuous nature of $y_{i,t}$. A jump in $y_{i,t}$ from $0$ to $1$ can lead to markedly different values of future probabilities $\p_{t+1}$. In particular, given the information set at time $t-1$, since $y_{i,t}=\1(u_{i,t}\leq p_{i,t})$, there is a discontinuity at $u_{i,t}=p_{i,t}$, which induces discontinuous and non-Lipschitz behavior in $\p_{t+1}$. The key to overcoming this difficulty lies in a suitable coupling construction.

\section{Aggregation}\label{sec_aggregation}

In this section, we study the behavior of the total number of successes at a given time, $X_t(N):=y_{1,t}+\ldots+y_{N,t}.$
While the previous section analyzed the limit of the system as $T \to \infty$ with $N$ fixed, we now investigate the complementary asymptotic regime in which $N \to \infty$. To initialize the dynamics, we consider an initial condition $\{\p_{1-\max\{s,q\}},\ldots,\p_{-1},\p_0\}.$
For $t=1-\max\{s,q\},\ldots,0$, we supplement this initial condition by setting $y_{i,t}=\1(u_{i,t}\leq p_{i,t})$, where the variables $u_{i,t}$ are $U[0,1]$ and independent of one another (and independent of all $u_{i,t}$ appearing in Eq.~\eqref{eq_dgp_y}). The initial probabilities $\p_{1-\max\{s,q\}},\ldots,\p_0$ are allowed to be random. In this case, all errors $u_{i,t}$ are assumed to be independent of the initial condition.

Our focus is on the case of rare events, where each $p_{i,t}$ tends to zero as $N\to\infty$. Studying such rare events is important because many applications involve outcomes that occur infrequently but whose aggregate frequency carries significant information. Examples include default counts in large credit portfolios, the number of extreme returns across assets, or spikes in systemic risk indicators. In these settings aggregation not only smooths out idiosyncratic noise but also highlights the limiting distributional behavior of the system, providing a tractable framework for inference and risk assessment.

\begin{figure}[t]
  \centering
  \includegraphics[width=0.8\linewidth]{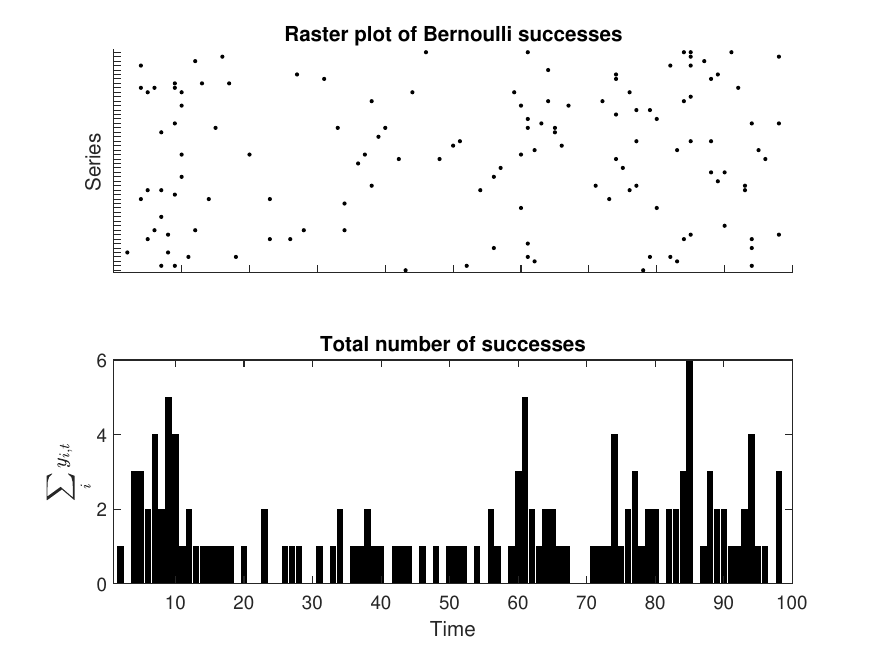}
\caption{Individual and aggregate successes for $N=50$ series generated from the interactive model \eqref{eq_multiple_k_diff} with $\omega_i=0.005$, $\alpha_i=0.01$, $\beta_i=0.6$, $\gamma_i=0.2$, and $p_{i,0}=1/38$ for all $i$. The top panel shows the raster of individual successes; the bottom panel reports the total number of successes at each time period. 
}
\label{fig_GAB_simulation_big}
\end{figure}

Figure \ref{fig_GAB_simulation_big} illustrates the behavior of the system when $N$ is large and successes are rare. Although the total number of series is $N = 50$, the aggregate number of successes remains small throughout the sample. By contrast, Figure \ref{fig_GAB_simulation_small} considers a setting with only $N = 5$ series but substantially higher success probabilities. Despite these differences in cross-sectional size and individual success rates, the total number of successes is of comparable magnitude across the two data-generating processes.

As we show below, aggregation naturally leads to a Poisson autoregression/INGARCH, providing a micro-level theoretical foundation for the applicability of such models. In the linear case, the conditional mean of the aggregate count follows a GARCH-type recursion, recovering exactly the model of \citet{ferland2006integer,fokianos2009poisson}. More generally, our nonlinear Theorem \ref{th_aggregation_nonlin_big} shows that richer micro-level interactions survive aggregation in a precise sense: the limiting intensity $\lambda_t$ inherits a genuinely nonlinear structure, yielding a broader class of Poisson autoregressions than the parametric nonlinear specifications considered in \citet[Section 2.3]{fokianos2009poisson}.

Let us define a linear Poisson autoregression of order $1$, $[X_t,\lambda_t]_{t=0,1,\ldots}$. Here $X_t\in \mathbb Z_{\ge 0}$, $\lambda_t\in \mathbb R_{\ge 0}$, and we denote by $\I^X_{t}$ the $\sigma$-field generated by $\{\lambda_0, X_0,\ldots,\lambda_{t},X_t,\lambda_{t+1}\}$. Then, for each $t\geq0$, $X_t$ follows a Poisson distribution, $Pois(\lambda_t)$, conditional on $\I^X_{t-1}$. For each $t> 0$, $\lambda_t$ satisfies for deterministic $c^P,\gamma^P,\beta^P\geq0$
\begin{equation}\label{eq_lambdat}
  \lambda_t=c^P+\gamma^P X_{t-1}+\beta^P\lambda_{t-1}.
\end{equation}
Dynamics \eqref{eq_lambdat} implies that $\sigma(\lambda_0, X_0,\ldots,\lambda_{t},X_t,\lambda_{t+1})=\sigma(\lambda_0, X_0,X_1,\ldots,X_t)$.

\begin{theorem}\label{th_aggregation}
Consider a multivariate model \eqref{eq_multiple_k_diff}. Suppose that $\omega_i\equiv\omega_{i,N}=\frac{c_i}{N}$, $\alpha_i\equiv\alpha_{i,N}=\frac{a_i}{N^{\kappa}}$, $\kappa>0$, $0\leq c_i,a_i\leq C<\infty$, $\beta_i\equiv\beta$, $i=1,\ldots,N$. Additionally suppose that there exist limits $\bar{c}:=\lim\limits_{N\to\infty}\frac1{N}\sum\limits_{i=1}^{N}c_i$ and $\bar{\gamma}:=\lim\limits_{N\to\infty}\frac1{N}\sum\limits_{i=1}^{N}\gamma_i$. Finally, suppose that the initial condition satisfies
\begin{equation}\label{eq_aggregation_pi0}
  \sum_{i=1}^N p_{i,0}\xrightarrow[N\to\infty]{d}\lambda_0,\qquad \sum_{i=1}^N \mathbb{E}p_{i,0}\xrightarrow[N\to\infty]{}\mathbb{E}\lambda_0,\qquad
  \max\limits_{i=1,\ldots,N} p_{i,0}\xrightarrow[N\to\infty]{p}0,
\end{equation}
where $\lambda_0$ is some random variable. Then as $N\to\infty$, $\left[\sum_{i=1}^N y_{i,t},\,\sum_{i=1}^N p_{i,t}\right]_{t=0,1,\dots}$ converges in finite-dimensional distributions to a Poisson autoregression of order $1$, $[X_t,\lambda_t]_{t=0,1,\dots}$ with dynamics \eqref{eq_lambdat}, $c^P=\bar{c}$, $\gamma^P=\bar{\gamma}$, $\beta^P=\beta$.
\end{theorem}

Theorem \ref{th_aggregation} shows that different specifications for $\p_t$ can generate identical aggregate dynamics. For example, suppose that in both models $\alpha_i\equiv0$, and in the first model $\gamma_i = \tfrac{1}{3}$ for all $i$, while in the second model $\gamma_i = \tfrac{2}{3}$ for half of the units and $\gamma_i = 0$ for the rest, with all other coefficients are identical across the two models. In both cases, the average coefficient is $\bar{\gamma} = \tfrac{1}{3}$, implying identical aggregate dynamics. However, at the panel level, the two models differ substantially: in the second model, half of the probabilities evolve deterministically given $\p_0$.

\begin{remark}
One can extend the model to allow for heterogeneous $\beta_i$, as long as $\beta_i=\beta+o(1)$.
\end{remark}

\begin{remark}
Condition $\alpha_i\to0$ as $N\to\infty$ is essential for the Poisson–limit result. Although the rate of convergence may vary across $i$, keeping $\alpha_i$ fixed would change the asymptotic behavior of the total sum $\sum_{i=1}^{N} y_{i,t}$. If $\alpha_i$ does not vanish, with positive probability we get $y_{i,t}=1$, so that $p_{i,t+1}\geq\alpha_i$. Therefore, at least one unit retains a non-negligible success probability, violating the small-probability condition required by the Poisson limit theorem. Given that $\alpha_i+\beta_i+\gamma_i<\bar{C}<1$ for all $i$, we expect that the effect of $y_{i,t}=1$ will decay over time. Consequently, as the system evolves, the aggregate process behaves approximately like a Poisson process with an additional serially correlated component induced by past feedback effects.

For example, if $\alpha_i\equiv\alpha>0$, $c_i\equiv c>0$, and $\beta=\gamma_i=0$, then, conditionally on $X_t=k$, at $t+1$ there are $k$ non-negligible probabilities $p_{i,t+1}=c/N+\alpha\approx\alpha$ and $N-k$ probabilities $p_{i,t+1}=c/N$. Thus, aggregating the outcomes of the former, we get Bin$(k,\alpha)$, while aggregating the outcomes of the latter, we get an independent Pois$(c)$ component, and
$$
X_{t+1}\stackrel{d}{=} \text{Pois}(c)+\text{Bin}(X_t,\alpha).
$$
\end{remark}
\begin{remark}
If $\lambda_0\thicksim \text{Gamma}(\mathfrak{m},b)$, where $\mathfrak{m}$ is a shape parameter and $b$ is a rate parameter, then the unconditional distribution of $X_0$ is negative binomial, $\text{NB}\left(\mathfrak{m},\tfrac{b}{1+b}\right)$. If additionally, $\bar{c}=\bar{\gamma}=0$, so that the evolution of the Poisson intensity \eqref{eq_lambdat} is deterministic given $\lambda_0$, then $\lambda_t\equiv\beta^t\lambda_0\thicksim\text{Gamma}(\mathfrak{m},b\beta^{-t})$. Thus, the unconditional distribution of $X_t$ is again negative binomial, $\text{NB}\left(\mathfrak{m},\tfrac{b}{b+\beta^{t}}\right)$.
\end{remark}

There are many cases when the initial condition \eqref{eq_aggregation_pi0}  is satisfied. For example, if $p_{i,0}=\frac{r_i}{N}$, where $r_i$ are i.i.d.~with finite mean. Then ${\frac1{N}\max_{i=1,\ldots,N}r_i\xrightarrow[N\to\infty]{p}0}$ and ${\sum_{i=1}^N p_{i,0}=\frac1{N}\sum_{i=1}^N r_{i}\xrightarrow[N\to\infty]{p}\mathbb{E}r_1}$. Moreover, it is satisfied by the stationary solution, as shown below.
\begin{proposition}\label{rem_init_cond_aggreg}
If $\beta+\gamma_i<\bar{C}<1$ for all $i$, then the unique stationary solution of the multivariate model \eqref{eq_multiple_k_diff} satisfies assumption on the initial condition \eqref{eq_aggregation_pi0}.
\end{proposition}

Theorem \ref{th_aggregation} can be extended to allow for multiple lags and various nonlinearities. Consider a nonlinear interactive model, which generalizes model \eqref{eq_multiple_k_diff}:
\begin{equation}\label{eq_multiple_k_diff_nonlin_big}
\begin{split}
p_{i,t}&=\frac{c_i}{N}+\frac1{N^{\kappa}}f_{\alpha,i}(y_{i,t-1},\ldots,y_{i,t-s})+\sum\limits_{\tau=1}^{s}\beta_{\tau} p_{i,t-\tau}\\
&+f_{\gamma,i}\left(\frac1{N}\sum\limits_{j=1}^{N}y_{j,t-1},\ldots,\frac1{N}\sum\limits_{j=1}^{N}y_{j,t-s}, \frac1{N}\sum\limits_{j=1}^{N}p_{j,t-1},\ldots,\frac1{N}\sum\limits_{j=1}^{N}p_{j,t-s}\right)\\
&+\frac1{N}\boldsymbol\gamma_i^{\T}
f_{\gamma}\left(\sum\limits_{j=1}^{N}y_{j,t-1},\ldots,\sum\limits_{j=1}^{N}y_{j,t-s}, \sum\limits_{j=1}^{N}p_{j,t-1},\ldots,\sum\limits_{j=1}^{N}p_{j,t-s}\right),
\end{split}\end{equation}
where $\kappa>0$, $\beta_{\tau}\geq0$, $0\leq c_i\leq C$, $\boldsymbol\gamma_i\in\mathbb R_+^k$ with $\|\boldsymbol\gamma_i\|_1\leq C$, $f_{\alpha,i}:\{0,1\}^s\to\mathbb{R}_+$, ${f_{\gamma,i}:[0,1]^{2s}\to\mathbb{R}_+}$, and $f_{\gamma}:\mathbb{R}^{2s}_+\to\mathbb{R}^{k}_+$. We assume
$$f_{\alpha,i}(0,\ldots,0)=f_{\gamma,i}(0,\ldots,0)=0,\qquad f_{\gamma}(0,\ldots,0)=0.$$
We introduce the nonlinear Poisson autoregression counterpart of the interactive model \eqref{eq_multiple_k_diff_nonlin_big}, $[X_t,\lambda_t]_{t>-s}$. Denote by $\I^X_{t}$ the $\sigma$-field generated by $\{\lambda_{-s+1},X_{-s+1}\ldots,\lambda_{t},X_t,\lambda_{t+1}\}$.  Then, for each $t>-s$, $X_t$ follows a Poisson distribution, $Pois(\lambda_t)$, conditional on $\I^X_{t-1}$. For each $t> 0$, $\lambda_t$ satisfies for deterministic $c^P\geq0,\,\boldsymbol{\gamma}^P\in\mathbb{R}^{k}_{+},\,\boldsymbol{\beta}^P\in\mathbb{R}^{2s}_{+}$
\begin{equation}\label{eq_lambdat_big}
  \lambda_t
  =c^P+(X_{t-1},\ldots,X_{t-s},\lambda_{t-1},\ldots,\lambda_{t-s})\boldsymbol{\beta}^P
  +\left[f_{\gamma}(X_{t-1},\ldots,X_{t-s},\lambda_{t-1},\ldots,\lambda_{t-s})\right]^{\T}{\boldsymbol\gamma}^{P}.
\end{equation}
Dynamics \eqref{eq_lambdat_big} implies that
$$\sigma(\lambda_{-s+1}, X_{-s+1},\ldots,\lambda_{t},X_t,\lambda_{t+1})=\sigma(\lambda_{-s+1}, X_{-s+1},\ldots,\lambda_0, X_0,X_1,\ldots,X_t).$$
\begin{theorem}\label{th_aggregation_nonlin_big}
Consider a nonlinear interactive model \eqref{eq_multiple_k_diff_nonlin_big}.
Suppose that $\kappa,C,A,L_\gamma,M_1,M_2$ are finite constants independent of $N$. Suppose that, for all $i$,
$$0\leq f_{\alpha,i}(y_1,\ldots,y_s)\leq A\sum_{\tau=1}^s y_\tau.$$
Suppose that, for all $i$, $f_{\gamma,i}$ is twice continuously differentiable and with uniformly bounded gradient and Hessian,
$$
\sup_{x\in[0,1]^{2s}}\|\nabla f_{\gamma,i}(x)\|_1\leq M_1,
\qquad
\sup_{x\in[0,1]^{2s}}\|\nabla^2f_{\gamma,i}(x)\|\leq M_2.
$$
Suppose that $f_\gamma$ is Lipschitz,
$$
\|f_\gamma(x)-f_\gamma(x')\|_\infty \leq L_\gamma\|x-x'\|_\infty.
\qquad x,x'\in\mathbb{R}_{+}^{2s},
$$
Suppose that $\sum_{\tau=1}^s\beta_\tau+M_1+CL_\gamma<1$ and there exist limits
$$\bar{c}:=\lim\limits_{N\to\infty}\frac1{N}\sum\limits_{i=1}^{N}c_i,\qquad \bar{\boldsymbol\gamma}:=\lim\limits_{N\to\infty}\frac1{N}\sum\limits_{i=1}^{N}\boldsymbol\gamma_i,\qquad \bar{\mathbf{f}}:=\lim\limits_{N\to\infty}\frac1{N}\sum\limits_{i=1}^{N}\nabla f_{\gamma,i}(0,\ldots,0).$$
Finally, suppose that the initial condition satisfy
\begin{equation*}\label{eq_aggregation_pi0_n_big}
  \sum_{i=1}^N p_{i,\tau}\xrightarrow[N\to\infty]{d}\lambda_{\tau},\qquad \sum_{i=1}^N \mathbb{E}p_{i,\tau}\xrightarrow[N\to\infty]{}\mathbb{E}\lambda_{\tau},\qquad \max\limits_{i=1,\ldots,N} p_{i,\tau}\xrightarrow[N\to\infty]{p}0,
\end{equation*}
where $(\lambda_{1-s},\ldots,\lambda_0)$ is a random vector and convergence is joint over $\tau=-s+1,\ldots,0$. Then as $N\to\infty$, $\left[\sum_{i=1}^N y_{i,t},\,\sum_{i=1}^N p_{i,t}\right]_{t=-s+1,-s+2,\dots}$ converges in finite-dimensional distributions to a nonlinear Poisson autoregression of order $s$, $[X_t,\lambda_t]_{t=-s+1,-s+2,\dots}$ with dynamics \eqref{eq_lambdat_big}, $c^P=\bar{c}$, $\boldsymbol\gamma^P=\bar{\boldsymbol\gamma}$, $\boldsymbol\beta^P=\bar{\mathbf{f}}+(0,\ldots,0,\beta_1,\ldots,\beta_s)^{\T}$.
\end{theorem}
Condition $\sum_{\tau=1}^s\beta_\tau+M_1+CL_\gamma<1$ implies that, starting from $N$ large enough, the right-hand side of \eqref{eq_multiple_k_diff_nonlin_big}
belongs to $[0,1]$ for every $i$ and $t$. Conditions $f_{\gamma,i}(\cdot)\geq0$ and $f_{\gamma,i}(0,\ldots,0)=0$ imply that the gradient $\nabla f_{\gamma,i}(0,\ldots,0)\in\mathbb{R}_+^{2s}$, so that $\bar{\mathbf{f}}\in\mathbb{R}_+^{2s}$.

When $s=k=1$, $f_{\alpha,i}(y)=a_i y$, $f_{\gamma,i}\equiv0$, $f_{\gamma}(y,p)=y$, Theorem \ref{th_aggregation_nonlin_big} reduces to Theorem \ref{th_aggregation}. Alternatively, one can set $f_{\gamma}\equiv0$ and $f_{\gamma,i}(y,p)=\gamma_i y$. The linear multi-lag extension of Theorem \ref{th_aggregation} can be obtained by setting $k=1$, $f_{\alpha,i}(y_1,\ldots, y_s)=\sum_{\tau=1}^{s}a_{i,\tau} y_{\tau}$, $f_{\gamma,i}(y,p)\equiv0$,
$f_{\gamma}(y_1,\ldots, y_s,p_1,\ldots, p_s)=\sum_{\tau=1}^{s}(\delta_{\tau}^y y_{\tau}+\delta_{\tau}^p p_{\tau})$, where $\delta_{\tau}^y,\delta_{\tau}^p\in\mathbb{R}_+$. Alternatively, one can set $f_{\gamma}\equiv0$ and $f_{\gamma,i}(y_1,\ldots, y_s,p_1,\ldots, p_s)=\gamma_i \sum_{\tau=1}^{s}(\delta_{\tau}^y y_{\tau}+\delta_{\tau}^p p_{\tau})$. Both representations recover the same limiting intensity $\lambda_t$, because for linear functions the prefactor $1/N$ can be taken out of the argument.

The term $\frac1{N^{\kappa}}f_{\alpha,i}(y_{i,t-1},\ldots,y_{i,t-s})$ is a nonlinear and multi-lag generalization of $\frac{a_i}{N^{\kappa}}y_{i,t-1}$. For example, it allows for across-time interactions $y_{i,t-1}y_{i,t-2}$. This term disappears in the limit due to the prefactor $N^{-\kappa}$. Without the prefactor the term can occasionally take non-negligible values, violating the requirement of the Poisson approximation that all probabilities are small. The term $f_{\gamma,i}(\cdot)$ is a nonlinear and multi-lag generalization of the interactions, which are governed by the average number of successes. The averages are $o_p(1)$, leading to the average gradient of $f_{\gamma,i}$ in the limit as the outcome of the first-order Taylor approximation. Thus, these nonlinearities asymptotically become linear. Finally, the term $\frac1{N}\boldsymbol\gamma_i^{\T}f_{\gamma}(\cdot)$ is the one which remains nonlinear in the limit, due to its arguments converging to non-degenerate random variables. Unlike the $f_{\alpha,i}(\cdot)$ and $f_{\gamma,i}(\cdot)$ terms, where individual heterogeneity enters through $i$-indexed coefficients that average out in the limit, $f_{\gamma}(\cdot)$ is common across all $i$: every coordinate $i$ is exposed to the same aggregate quantity $f_{\gamma}\left(\sum\limits_{j=1}^{N}y_{j,t-1},\ldots,\sum\limits_{j=1}^{N}y_{j,t-s}, \sum\limits_{j=1}^{N}p_{j,t-1},\ldots,\sum\limits_{j=1}^{N}p_{j,t-s}\right)$, with individual heterogeneity entering only through $\boldsymbol\gamma_i$.

\medskip

A distinctive feature of the aggregation in Theorems \ref{th_aggregation} and \ref{th_aggregation_nonlin_big} is that each coordinate $i$ interacts with the equal-weight cross-sectional aggregates $\sum\limits_{j=1}^{N}y_{j,t}$ and $\sum\limits_{j=1}^{N}p_{j,t}$, leaving no room for heterogeneous interaction weights. In some special cases, such as the network GAB model, this restriction can be relaxed, as illustrated in Theorem \ref{th_aggregation_network} below.

\begin{theorem}\label{th_aggregation_network}
Consider a multivariate network model \eqref{eq_multiple_k_diff_W}. Suppose that the network matrix $W$ is doubly stochastic and all out-degrees satisfy $d_i^{out}/\log N\geq d, d>0$, where $d_i^{out}:=\sum\limits_{j=1}^{N} \1(W_{ij}\neq0)$. Suppose that $\omega_i\equiv\omega_{i,N}=\frac{c_i}{N}$, $\alpha_i\equiv\alpha_{i,N}=\frac{a_i}{N^{\kappa}}$, $\kappa>0$, $0\leq c_i,a_i\leq C<\infty$, $\beta_i\equiv\beta$,  $\gamma_i\equiv\gamma$, $i=1,\ldots,N$, and there exists limit $\bar{c}:=\lim\limits_{N\to\infty}\frac1{N}\sum\limits_{i=1}^{N}c_i$. Finally, suppose that the initial condition satisfies
\begin{equation*}
  \sum_{i=1}^N p_{i,0}\xrightarrow[N\to\infty]{d}\lambda_0,\qquad \sum_{i=1}^N \mathbb{E}p_{i,0}\xrightarrow[N\to\infty]{}\mathbb{E}\lambda_0,\qquad
  \max\limits_{i=1,\ldots,N} p_{i,0}\xrightarrow[N\to\infty]{p}0,
\end{equation*}
where $\lambda_0$ is some random variable. Then as $N\to\infty$, $\left[\sum_{i=1}^N y_{i,t},\,\sum_{i=1}^N p_{i,t}\right]_{t=0,1,\dots}$ converges in finite-dimensional distributions to a Poisson autoregression of order $1$, $[X_t,\lambda_t]_{t=0,1,\dots}$ with dynamics \eqref{eq_lambdat}, $c^P=\bar{c}$, $\gamma^P=\gamma$, $\beta^P=\beta$.
\end{theorem}

A row-normalized network matrix $W$ is doubly stochastic if, for example, for each $i$, $d_i^{in}=d_i^{out}=d$, where $d_i^{in}=\sum\limits_{j=1}^{N} \1(W_{ji}\neq0)$, $d_i^{out}=\sum\limits_{j=1}^{N} \1(W_{ij}\neq0)$ are in- and out-degrees of a vertex $i$. That means, each individual has $d$ incoming and $d$ outgoing links, i.e. the network is $d$-regular. The doubly stochastic structure, combined with homogeneity of $\gamma_i=\gamma$, ensures that the network weights wash out in the limit. Thus, the limiting intensity equation in Theorem \ref{th_aggregation_network} is identical to that of Theorem \ref{th_aggregation} with $\gamma_i=\gamma$.

The condition on the out-degree growing with $N$ ensures that terms $\sum_j W_{ij}y_{j,t-1}$ are of order $o_p(1)$, mimicking the scaling of the equal-weight average $\frac1{N}\sum\limits_{j=1}^{N}y_{j,t}$ and thereby ensuring the Poisson approximation applies.

\section{Estimation}\label{sec_estimation}
We begin by discussing estimation when success probabilities are bounded away from $0$ and $1$, the regime where traditional maximum likelihood methods apply, such as
\citet{moysiadis2014binary,davis2016theory,fokianos2017binary,aknouche2021count,lee2023modeling}. We then briefly discuss what can be done in rare-event scenarios.

Suppose that the success probability function $\g(\cdot)$ depends on an unknown parameter $\theta_0 \in \mathbb{R}^m$, such as the coefficients in the models discussed in Section \ref{sec_examples}. Suppose time periods $t=1,\ldots,T$ are observed. Then the parameter $\theta_0$ can be estimated by maximum likelihood. Since the innovations $u_{i,t}$ are independent and uniformly distributed on $[0,1]$, the likelihood and log-likelihood functions conditional on the initial values $\y_0,\ldots,\y_{-q+1},\p_0,\ldots,\p_{-s+1}$ take the form
\begin{equation*}\label{eq_likelihood}
\begin{split}
  \mathcal{L}_{T}(\theta)
  =&  \prod_{t=1}^{T}\prod_{i=1}^{N} (\tilde{p}_{i,t})^{y_{i,t}} (1-\tilde{p}_{i,t})^{1-y_{i,t}},\qquad
  \text{where }\tilde{\p}_t \text{ is recursively defined by}\\
  \tilde{\p}_t=&\begin{cases}
    \g\left(\{\tilde{\p}_{\tau}\}_{\tau=t-1,\ldots,t-s},\,\{\y_{\tau}\}_{\tau=t-1,\ldots,t-q};\theta\right), & t>0,\\
    \p_t, & t\leq0,
  \end{cases}
\end{split}\end{equation*}
\begin{equation*}\label{eq_loglikelihood}
\begin{split}
  \mathcal{Q}_T(\theta):=\log\mathcal{L}_{T}(\theta)&
  = \sum_{t=1}^{T}\sum_{i=1}^{N} \left[y_{i,t}\log \tilde{p}_{i,t}+ (1-y_{i,t})\log\left(1- \tilde{p}_{i,t}\right)\right].
\end{split}\end{equation*}
In empirical applications the first $q$ realizations of the binary process $\y_t$ can be used as initial values. The corresponding probabilities $\p_t$, however, are unobserved. A practical approach is, therefore, to initialize them using the sample means of $\y_t$. In this section we impose Lipschitz-type contraction of Assumptions \ref{ass_contraction1}, \ref{ass_contraction2} (in which the constants are assumed to be $\theta$-independent). As shown in the proof of Theorem \ref{th_consistency}, this guarantees that the effect of the probability initialization on the normalized log-likelihood vanishes uniformly over $\theta$.

\begin{assumption}\label{ass_g_theta}~
\begin{enumerate}
  \item \label{enu:g_theta:theta} $\theta_0\in\Theta$ for some compact set $\Theta\subset \mathbb{R}^m$;
  \item \label{enu:g_theta:contin} $\g\left(\{\p_{\tau}\}_{\tau=t-1,\ldots,t-s},\,\{\y_{\tau}\}_{\tau=t-1,\ldots,t-q};\theta\right)$ is continuous in $\theta\in\Theta$;
  \item \label{enu:g_theta:bounded} $\exists\eps>0$ such that for all $i=1,\ldots,N$, $\theta\in\Theta$, $g_i(\cdot;\theta)\in[\eps,1-\eps]$;
  \item \label{enu:g_theta:unique} Let $\y_t$, $\p_t$ be the stationary solution of Theorem \ref{th_stationary_unique}. If $$Prob\left(\g\left(\{\p_{\tau}\}_{\tau=t-1,\ldots,t-s},\,\{\y_{\tau}\}_{\tau=t-1,\ldots,t-q};\theta\right)
      =\g\left(\{\p_{\tau}\}_{\tau=t-1,\ldots,t-s},\,\{\y_{\tau}\}_{\tau=t-1,\ldots,t-q};\theta_0\right)\right)=1,$$
      then $\theta=\theta_0$.
\end{enumerate}
\end{assumption}
In all linear GAB examples, we do not lose any generality by imposing the compactness condition, because $\g$ maps to $[0,1]^N$. For nonlinear $\g$, compactness of the parameter space can be relaxed if, with positive probability, $\g\left(\{\p_{\tau}\}_{\tau=t-1,\ldots,t-s},\,\{\y_{\tau}\}_{\tau=t-1,\ldots,t-q};\theta\right)$ goes to $0$ or $1$ as $\theta \to \infty$. This ensures that the log-likelihood is never maximized at unbounded values of $\theta$.

Assumption \ref{ass_g_theta}\ref{enu:g_theta:bounded} is introduces to keep the log-likelihood and its derivatives uniformly bounded. It resembles the standard GARCH condition that the intercept be bounded away from zero (see, e.g., \citet[Example 7.1]{francq2019garch}). In our setup, when the functions $g_i$ are linear, this means that all intercepts are separated from zero and that the sums of all coefficients in $g_i$ are separated from one. Note that if, for all $i$, $g_i(0,\ldots,0;\theta_0)=0$, then $y_{i,t}=p_{i,t}=0$ is a degenerate stationary solution. Similarly, if, for all $i$, $g_i(1,\ldots,1;\theta_0)=1$, then $y_{i,t}=p_{i,t}=1$ is a degenerate stationary solution. In both cases, the log-likelihood becomes infinite. By imposing \ref{ass_g_theta}\ref{enu:g_theta:bounded}, we rule out such degenerate settings.

Assumption \ref{ass_g_theta}\ref{enu:g_theta:unique} requires the existence of a unique parameter value $\theta$ that is compatible with the data, i.e., identification of $\theta_0$.

The maximum likelihood estimator solves
$$\hat{\theta}^{MLE}=\arg\max\limits_{\theta\in\Theta} \mathcal{Q}_T(\theta).$$
\begin{theorem}\label{th_consistency}
Suppose that the multivariate binary time series model \eqref{eq_dgp_y}--\eqref{eq_dgp_p} satisfies Assumption \ref{ass_g_theta} and either Assumption \ref{ass_contraction1} or \ref{ass_contraction2}, in which the constants are assumed to be $\theta$-independent. Then $\hat{\theta}^{MLE}\xrightarrow[T\to\infty]{p}\theta_0$.
\end{theorem}
Assumption \ref{ass_contraction1} or \ref{ass_contraction2} is imposed to ensure ergodicity and uniqueness of a stationary solution.

\begin{assumption}\label{ass_g_theta2}~
\begin{enumerate}
  \item \label{enu:g_theta2:theta_int} $\theta_0$ is an interior point of $\Theta\in \mathbb{R}^m$;
  \item \label{enu:g_theta2:differ} 
      There exists a compact neighborhood $\Theta_0\subset \Theta$ of $\theta_0$ such that, for every fixed $(\y_{t-1},\ldots,\y_{t-q})$, the function $\g$ is twice continuously differentiable jointly in $(\p_{t-1},\ldots,\p_{t-s},\theta)$ on  $[0,1]^{Ns}\times\Theta_0$;
  \item \label{enu:g_theta2:H}
 The matrix
  $$ H_0=\mathbb{E}\sum_{i=1}^N
      \frac{1}
           {g_{i,t,\theta_0}(1-g_{i,t,\theta_0})}\left(\frac{d}{d\theta}g_{i,t,\theta_0}\right)\left(\frac{d}{d\theta}g_{i,t,\theta_0}\right)^{\T}$$
  is nonsingular, where
  $$g_{i,t,\theta_0}:=g_i\left(\{\p_{\tau}\}_{\tau=t-1,\ldots,t-s},\,\{\y_{\tau}\}_{\tau=t-1,\ldots,t-q};\theta_0\right)$$
  and $\frac{d}{d\theta}g_{i,t,\theta_0}$ denotes the total derivative evaluated at $\theta_0$, accounting for the dependence of the lagged probabilities on $\theta$ through the recursion. More precisely, if
  $D_t\in\mathbb R^{N\times m}$ has $i$th row  $(d g_{i,t,\theta_0}/d\theta)^{\T}$, then $D_t$ is the unique bounded stationary solution of the recursion
  \begin{equation*}
      D_t  = G_{\theta,t} +\sum_{j=1}^s G_{p_j,t}D_{t-j},
  \end{equation*}
  where
  $$ G_{\theta,t}:= \left.\frac{\partial}{\partial\theta^{\T}} \g\left( \{\p_{\tau}\}_{\tau=t-1,\ldots,t-s},\{\y_{\tau}\}_{\tau=t-1,\ldots,t-q};\theta\right) \right|_{\theta=\theta_0},$$
  and
  $$ G_{p_j,t}:=\left. \frac{\partial}{\partial\p_{t-j}^{\T}} \g\left(\{\p_{\tau}\}_{\tau=t-1,\ldots,t-s},\{\y_{\tau}\}_{\tau=t-1,\ldots,t-q};\theta\right) \right|_{\theta=\theta_0}.$$
\end{enumerate}
\end{assumption}

\begin{theorem}\label{th_asy_normality}
Suppose that the multivariate binary time series model \eqref{eq_dgp_y}--\eqref{eq_dgp_p} satisfies Assumptions \ref{ass_g_theta}, \ref{ass_g_theta2}, and either \ref{ass_contraction1} or \ref{ass_contraction2}, in which the constants are assumed to be $\theta$-independent. Then $\sqrt{T}(\hat{\theta}^{MLE}-\theta_0)\xrightarrow[T\to\infty]{d}\mathcal{N}(0,H_0^{-1})$.
\end{theorem}

Note that the boundedness condition in Assumption \ref{ass_g_theta}\ref{enu:g_theta:bounded} rules out rare events, where $p_{i,t}$ is close to zero. Both consistency and asymptotic normality rely on this condition: it ensures the log-likelihood is well-defined and uniformly dominated, and keeps the score and Hessian bounded. As we illustrated in Section \ref{sec_aggregation}, when $N$ is large and $p_{i,t}$ is small, the aggregate count $\sum_{i=1}^N y_{i,t}$ is approximately Poisson with mean $\sum_{i=1}^N p_{i,t}$. Thus, a natural alternative for estimation in such settings is the Poisson autoregressive model, which replaces the binary observation equation with a Poisson likelihood.

In the framework of Theorem \ref{th_aggregation}, when $\beta<1$, $\gamma_i=\nu c_i$ for all $i$ (that is, when some form of homogeneity is present), the following intuition can be used to identify the Poisson log-likelihood in the MLE maximization in \eqref{eq_loglikelihood}.
First, in line with $\max_i p_{i,t}\xrightarrow[N\to\infty]{p}0$, we expect
$$\left|\sum_{i}y_{i,t}\log(1-g_{i,t,\theta})\right|\leq|\log(1-\max_i g_{i,t,\theta})|\sum_{i}y_{i,t}\xrightarrow[N\to\infty]{p}0.$$
Therefore,
\begin{equation*}\begin{split}
&\sum\limits_{i=1}^{N} y_{i,t}\log g_{i,t,\theta}+(1-y_{i,t})\log(1-g_{i,t,\theta})\\
&=\sum\limits_{i=1}^{N}\log(1-g_{i,t,\theta})+\sum\limits_{i=1}^{N}y_{i,t}\log g_{i,t,\theta}-\sum\limits_{i=1}^{N}y_{i,t}\log(1-g_{i,t,\theta})\\
&\approx -\sum\limits_{i=1}^{N} g_{i,t,\theta}+\left(\sum\limits_{i=1}^{N}y_{i,t}\right)\log\left(\sum\limits_{j=1}^{N}g_{j,t,\theta}\right)
+\sum\limits_{i=1}^{N}y_{i,t}\log\left(g_{i,t,\theta}/\sum\limits_{j=1}^{N}g_{j,t,\theta}\right),
\end{split}\end{equation*}
so that the first two terms operate on aggregate quantities and exactly match the Poisson log-likelihood. The last term is close to being independent of aggregate quantities and can therefore be maximized separately. To see this, applying iterative back-substitution yields
$$p_{i,t}=\frac{c_i}{N(1-\beta)}+\beta^t\left(p_{i,0}-\frac{c_i}{N(1-\beta)}\right)+\frac{a_i}{N^{\kappa}}\sum\limits_{\tau=0}^{t-1}\beta^{\tau}y_{i,t-1-\tau}
+\frac{\nu c_i}{N}\sum\limits_{\tau=0}^{t-1}\beta^{\tau}\sum\limits_{j=1}^Ny_{j,t-1-\tau},$$
\begin{equation*}\begin{split}
\sum_i p_{i,t}&=\frac{\sum_i c_i}{N(1-\beta)}+\beta^t\left(\sum_i p_{i,0}-\frac{\sum_i c_i}{N(1-\beta)}\right)\\
&+\sum_i\frac{a_i}{N^{\kappa}}\sum\limits_{\tau=0}^{t-1}\beta^{\tau}y_{i,t-1-\tau}
+\left(\frac{\nu\sum_i c_i}{N}\right)\sum\limits_{\tau=0}^{t-1}\beta^{\tau}\sum\limits_{j=1}^Ny_{j,t-1-\tau}.
\end{split}\end{equation*}
When $t$ is large, $\beta^t\approx 0$, and the second terms in both expressions above are negligible. Since $\kappa>0$, the third terms are also negligible, leaving
$$\frac{g_{i,t,\theta}}{\sum\limits_{j=1}^{N}g_{j,t,\theta}}
\approx\frac{\frac{c_i}{N}\left(\frac1{1-\beta}+\nu\sum\limits_{\tau=0}^{t-1}\beta^{\tau}\sum\limits_{j=1}^Ny_{j,t-1-\tau}\right)}
{\frac{\sum_j c_j}{N}\left(\frac1{1-\beta}+\nu\sum\limits_{\tau=0}^{t-1}\beta^{\tau}\sum\limits_{j=1}^Ny_{j,t-1-\tau}\right)}=\frac{c_i}{\sum_j c_j}.$$
That is, only the relative parameters $\tilde{c}_i:=c_i/\sum_j c_j$ remain. These lie on a simplex and are independent of the aggregate values, so the original maximization problem splits into two separate problems, the first of which corresponds to a Poisson likelihood.

\section{Covariates}\label{sec_covariates}

The model can be extended to include exogenous covariates $\z_t\in\mathbb{R}^k$, where the covariate process $\{\z_t\}$ is independent of the innovation process $\{\mathbf{u}_t\}$, so that
$$p_{i,t}=g_i\left(\{\p_{\tau}\}_{\tau=t-1,\ldots,t-s},\,\{\y_{\tau}\}_{\tau=t-1,\ldots,t-q},\,\z_t\right).$$
Provided $\z_t$ is strictly stationary and ergodic, Theorems \ref{th_stationary_sol} and \ref{th_stationary_unique} hold unchanged under Assumptions \ref{ass_contraction1}--\ref{ass_contraction2} as stated. This follows because in the coupling construction both trajectories share the same realization of $\z_t$, so it
cancels identically from the Lipschitz bounds \eqref{eq_g_contraction1} and \eqref{eq_g_contraction2}. Similarly, estimation Theorems \ref{th_consistency} and \ref{th_asy_normality} continue to hold.

The more interesting question is how to incorporate covariates in the aggregation framework. If $k$ is fixed, then $\z_t$ must enter each equation for $p_{i,t}$ with a coefficient vanishing in $N$. Otherwise, the finite order of $\z_t$ would prevent $p_{i,t}$ from being of order $1/N$. Paralleling the treatment of the intercept $\omega_i$, we therefore add a linear term $\frac{1}{N}\delta_i^{\T}\z_{t}$ to each $p_{i,t}$ equation in interactive and network GAB models. Letting $\bar{\delta}:=\lim_{N\to\infty}\frac1{N}\sum_{i=1}^{N}\delta_i$, this contributes an extra term $\bar{\delta}^{\T}\z_t$ to the right-hand side of the evolution equation for $\lambda_t$.

If $k$ grows with $N$, the covariates must themselves aggregate to a limiting quantity. We decompose $\z_t=(\z_t^{com},\z_t^{ind})$, where $\z_t^{com}\in\mathbb{R}^{r_c}$ is a common component of fixed dimension affecting all units, and $\z_t^{ind}=\left((\z_{1,t}^{ind})^{\T},\ldots,(\z_{N,t}^{ind})^{\T}\right)^{\T}\in\mathbb{R}^{Nr}$ stacks unit-specific components with $\z_{i,t}^{ind}\in\mathbb{R}^r$, so that $k=Nr+r_c$. The common component is handled exactly as in the fixed $k$ case. For the individual components, $p_{i,t}$ may depend on $\z_{i,t}^{ind}$ directly and on the cross-sectional average $\frac{1}{N}\sum_{j=1}^N\z_{j,t}^{ind}$ (or the weighted average $\sum_{j=1}^N W_{ij}\z_{j,t}^{ind}$ in the network GAB case, as in Theorem \ref{th_aggregation_network}). The direct effect of $\z_{i,t}^{ind}$ on $p_{i,t}$ must be scaled by $1/N$, for the same order-$1/N$ reason as above, and must appear with a homogeneous coefficient $\eta$, so that the aggregation yields $\eta\frac1{N}\sum_{i=1}^N\z_{i,t}^{ind}$. The cross-sectional average must carry an additional factor of $1/N$, since unlike $y_{i,t}$ and $p_{i,t}$, the scale of the exogenous $\z_{i,t}^{ind}$ is not reduced by aggregation; the corresponding term in $p_{i,t}$ is therefore $\frac{\zeta_i}{N}\frac{1}{N}\sum_{j=1}^N\z_{j,t}^{ind}$ (or nonlinear bounded\footnote{Boundedness guarantees that $p_{i,t}\in[0,1]$ and that $\mathbb{E}f_{z}\left(\frac{1}{N}\sum_{j=1}^N\z_{j,t}^{ind}\right)\xrightarrow[N\to\infty]{}\mathbb{E}f_{z}\left(\bar{\z_{t}}\right)$.} function $\frac{\zeta_i}{N}f_{z}\left(\frac{1}{N}\sum_{j=1}^N\z_{j,t}^{ind}\right)$, so that $f_z$ survives in the limit analogously to the $f_{\gamma}$ term in
Theorem \ref{th_aggregation_nonlin_big}). Thus, assuming that $\frac{1}{N}\sum_{i=1}^N\z_{i,t}^{ind}\xrightarrow[N\to\infty]{p}\bar{\z}_t$,  $\mathbb{E}\frac{1}{N}\sum_{i=1}^N\z_{i,t}^{ind}\xrightarrow[N\to\infty]{}\mathbb{E}\bar{\z}_t<\infty$, and $\max_i\|\z_{i,t}^{ind}\|/N\xrightarrow[N\to\infty]{p}0$ for each $t$, the aggregated covariate $\bar{\z}_t$ appears on the right-hand side of the limiting evolution equation for $\lambda_t$.

\section{Empirical illustration}\label{sec_illustr}



In this section, we apply the interactive GAB model to study periods of unusually low stock returns among the constituents of the S$\&$P 100. Such episodes tend to coincide with market stress and elevated financial risk. We use binary indicators of tail events to study how downside risks emerge, cluster over time, and propagate across assets, defining a tail event for stock $i$ as its return falling below the $5$th percentile of its empirical distribution.

Our sample consists of daily returns for $94$ S$\&$P 100 constituents available over the full period 01.01.2012--31.12.2025, yielding 3,517 time series observations. The $94$ stocks are those for which data are available throughout the entire span; we also retain only one of the two Google share classes. Further details on the sample composition are provided in Section \ref{append_data} of the Appendix.

We subtract the risk-free rate and regress returns on the Fama-French five factors, then work with the idiosyncratic residuals. The Fama-French returns, which are reported in percent, are converted to decimal units before the regression. The final three years of data are reserved for forecast evaluation, so estimation uses observations up to 01.01.2023, giving $T=$2,768. The factor regressions are estimated on this estimation sample only. The $5$th percentile threshold for each stock is also computed on this estimation sample, and $y_{i,t} = 1$ whenever stock $i$'s idiosyncratic return falls below its threshold, and $y_{i,t} = 0$ otherwise. Figure \ref{fig_SP100_raster} summarizes the resulting binary panel. As expected, tail events cluster around periods of market stress, including COVID and the Fed's monetary policy tightening. This clustering motivates the inclusion of the cross-sectional interaction term as in Eq.~\eqref{eq_multiple_k_diff}, which captures the market-wide stress and allows individual tail probabilities to adjust.

\begin{figure}[t]
  \centering
  \includegraphics[width=0.9\linewidth]{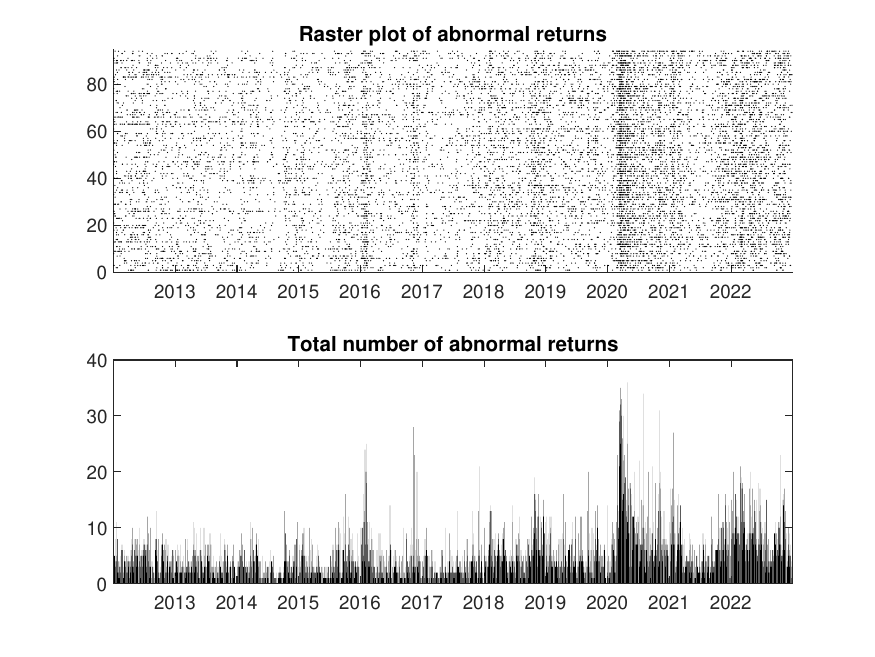}
\caption{Abnormally low (below $5\%$) returns in S$\&$P 100.}
\label{fig_SP100_raster}
\end{figure}

We first estimate the interactive GAB model \eqref{eq_multiple_k_diff}. The resulting estimates are shown in Figure \ref{fig_SP100_GAB_estim}. In line with the aggregation framework, both $\omega_i$ and $\alpha_i$ are relatively small, with $\operatorname{mean}(\omega_i)=0.004$ and $\operatorname{mean}(\alpha_i)=0.03$. The averages of the remaining coefficients are $\operatorname{mean}(\gamma_i)=0.18$ and $\operatorname{mean}(\beta_i)=0.71$.

\begin{figure}[t]
  \centering
  \includegraphics[width=0.6\linewidth]{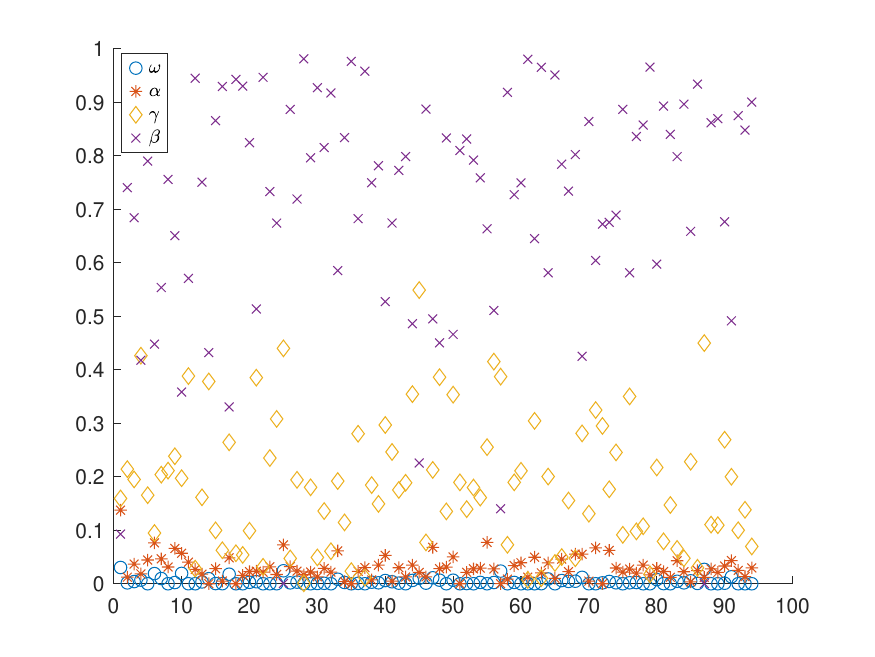}
\caption{Interactive GAB estimates for $94$ stocks.}
\label{fig_SP100_GAB_estim}
\end{figure}

For comparison with the aggregation results, we also re-estimate the model imposing $\alpha_i=0$, since this coefficient vanishes under aggregation, and compare the results with the Poisson MLE based on the aggregated data. The results are summarized in Table \ref{table_estim_res}. The GAB and Poisson estimates display a similar persistence pattern, although finite-$N$ differences remain, particularly for the intercept. When $\alpha_i=0$ is imposed, the intercept absorbs the mean contribution previously attributed to $\alpha_i$, leading to $\sum_i\omega_i$ larger than $\bar{c}$. The comparison is further illustrated by the filtered Poisson intensity $\lambda_t$ and the sum of GAB probabilities $\sum_i p_{i,t}$, shown in Figure \ref{fig_GAB_Poisson_insample}.


\begin{table}[t]
\centering
\begin{tabular}{lccccc}
\hline
 & $\sum_i\omega_i$ & $\operatorname{mean}(\omega_i)$ & $\operatorname{mean}(\alpha_i)$ & $\operatorname{mean}(\gamma_i)$ & $\operatorname{mean}(\beta_i)$  \\
\hline
Interactive GAB                    & 0.37 & 0.004 & 0.03 & 0.18 & 0.71  \\
Interactive GAB with $\alpha_i=0$ & 0.51 & 0.005 & -- & 0.24 & 0.65  \\
\hline
\hline
 & $\bar{c}$ &  & & $\bar{\gamma}$ & $\beta$  \\
\hline
Aggregated Poisson & 0.2 & -- & -- & 0.2 & 0.75  \\
\hline
\hline
\end{tabular}
\caption{Estimation Results}\label{table_estim_res}
\end{table}

\begin{figure}[t]
\centering
\begin{subfigure}[t]{0.48\linewidth}
    \centering
    \includegraphics[width=1.1\linewidth]{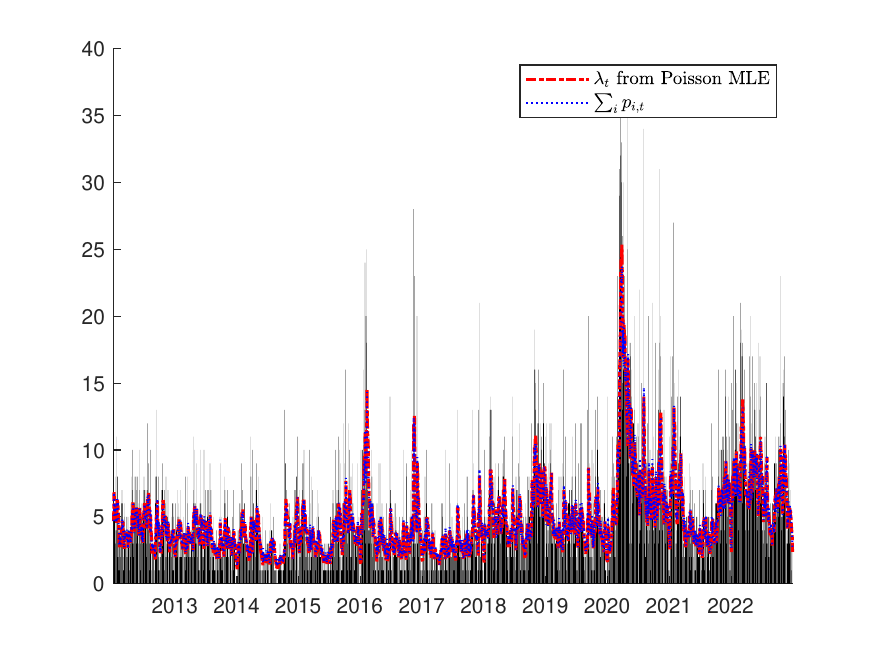}
    \caption{In-sample fit.}
    \label{fig_GAB_Poisson_insample}
\end{subfigure}
\hfill
\begin{subfigure}[t]{0.48\linewidth}
    \centering
    \includegraphics[width=1.1\linewidth]{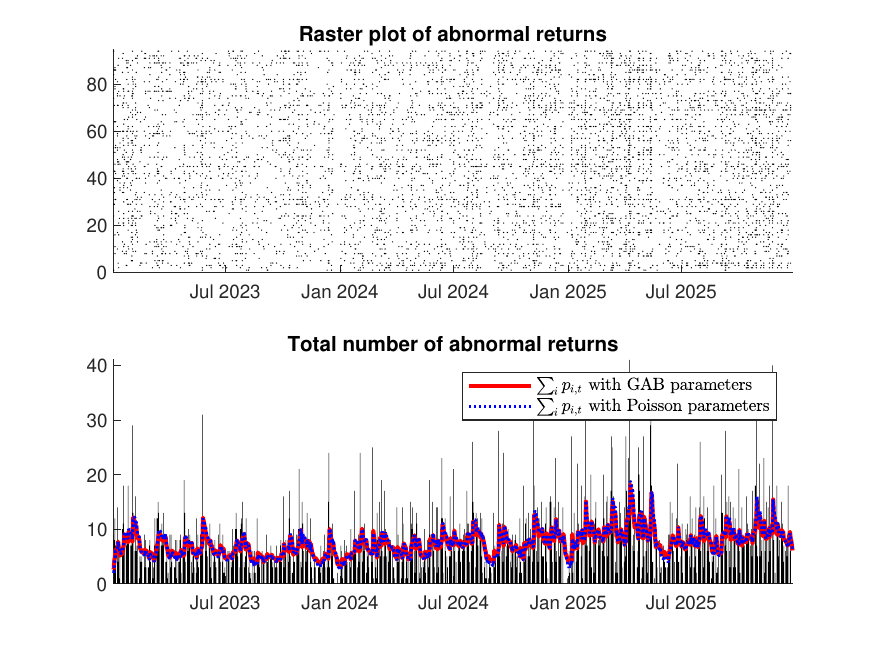}
    \caption{Out-of-sample fit.}
    \label{fig_GAB_Poisson_outofsample}
\end{subfigure}
\caption{Aggregate GAB probabilities $\sum_i p_{i,t}$ and Poisson intensity $\lambda_t$. The in-sample GAB path uses the restricted specification with $\alpha_i=0$ for comparison with the aggregation result, whereas the out-of-sample GAB path uses the unrestricted specification.}
\label{fig_GAB_Poisson}
\end{figure}

Finally, we compare the forecasting performance of the GAB model using the three-year sample that was excluded from the estimation stage. The out-of-sample binary data are constructed by comparing returns to the same $5$th percentile threshold used previously and are shown in Figure \ref{fig_GAB_Poisson_outofsample}. The out-of-sample period features a higher fraction of abnormally low returns than the in-sample $5\%$ benchmark ($8.7\%$ of the data points), indicating elevated uncertainty. The tail events nevertheless remain dispersed across stocks: even on the most concentrated out-of-sample day, $43.6\%$ of stocks are classified as tail events. We generate recursive 1-step-ahead forecasts of $y_{i,t}$ over the out-of-sample period ($T+1,\ldots,T+749$, three years of trading days) based on four models.

The first model is the unrestricted interactive GAB specification (heterogeneous parameters). The second is an interactive GAB model with homogeneous parameters calibrated from the Poisson maximum likelihood estimates, i.e., $\omega_i = \bar{c}/N$, $\alpha_i=0$, $\gamma_i = \bar{\gamma}$, and $\beta_i = \beta$. The third is a constant forecast $\hat{y}{i,t} = 0.05$, reflecting the unconditional $5\%$ probability of observing a success. The fourth is a persistence-based forecast $\hat{y}{i,t} = y_{i,t-1}$.

We report the mean squared errors of these forecasts in Table \ref{table_MSE_forecast}. The results show that the first model, the unrestricted interactive GAB with heterogeneous parameters, achieves the best forecasting performance. Although the GAB and Poisson estimation results in Table \ref{table_estim_res} are close, the second (homogeneous) model performs slightly worse. A Diebold--Mariano test based on the time series of daily cross-sectional mean squared losses, with a Newey--West long-run variance estimator using five lags, gives $DM=-2.08$ and $p=0.004$ for Model 1 versus Model 2, so the forecasting improvement is statistically significant. Both naive benchmarks are substantially outperformed.\footnote{This is consistent with the MSE results. Suppose that the observations are i.i.d.~Bernoulli. Then, for the fourth model, the squared forecast error is equal to $0$ when two consecutive realizations coincide and $1$ when they differ. The latter occurs with probability $2p(1-p) \approx 0.16$ for $p \approx 0.087$. For the third model, the squared forecast error takes values $0.05^2$ and $0.95^2$, and $MSE\approx 0.913\cdot0.05^2+0.087\cdot0.95^2\approx0.08$. For the first two models, $MSE\approx0.079$ is close to the optimal value, which corresponds to the forecast $\hat{y}_{i,t}\approx0.087$ yielding $MSE=p(1-p)^2+(1-p)p^2\approx0.079$.}

Figure \ref{fig_GAB_Poisson_outofsample} also presents the paths of $\sum_i p_{i,t}$ for the first two models. For the second model, this path coincides with the trajectory of the Poisson intensity $\lambda_t$. The superior forecasting performance of Models 1 and 2 relative to the naive benchmarks highlights the importance of the interaction term $\frac{1}{N}\sum_i y_{i,t}$ in capturing overall market conditions, a term absent from Models 3 and 4.

\begin{table}[t]
\centering
\begin{tabular}{cccc}
\hline
Model 1 & Model 2 & Model 3 & Model 4  \\
\hline
0.0789 & 0.0791 & 0.0805 & 0.154\\
\hline
\hline
\end{tabular}
\caption{Out-of-sample MSE (averaged across stocks and the out-of-sample period). Model 1: unrestricted interactive (heterogeneous) GAB. Model 2: Poisson (homogeneous GAB). Model 3: $5\%$ probability of success. Model 4: tomorrow is the same as today.}\label{table_MSE_forecast}
\end{table}

\section{Conclusion}\label{sec_conclusion}

This paper develops a general class of multivariate binary autoregressive models, in which each observation is a Bernoulli variable whose success probability evolves
according to a GARCH-type recursion driven by past outcomes and past probabilities. The framework accommodates nonlinear dynamics, cross-sectional interactions, and network-weighted dependence structures.

Three sets of results are established. The bounded support of binary outcomes and probabilities guarantees, without any parameter restrictions, the existence of a strictly stationary solution, a feature with no counterpart in continuous-valued GARCH models, while uniqueness and ergodicity follow from a coupling argument that exploits a Lipschitz contraction condition on the probability update function. The most distinctive results concern aggregation: under a rare-events scaling in which individual success probabilities are of order $1/N$, the aggregate number of successes converges to a Poisson autoregression, providing a rigorous micro-foundation for this commonly used model. The network extension shows further that, under sufficient connectivity, heterogeneous interaction weights wash out asymptotically, leaving a limiting intensity equation that does not depend on the network topology. Finally, in the moderate-probability regime, the MLE is consistent and asymptotically normal with asymptotic variance equal to the inverse Fisher information; in the rare-events regime, the Poisson likelihood applied to the aggregate count is the natural substitute, and the aggregation theory supplies the mapping from its parameters back to those of the underlying binary system.

Several directions remain open. On the theoretical side, a unified asymptotic theory in which both $T$ and $N$ grow simultaneously would constitute a valuable generalization. The extension of asymptotic normality for the MLE to settings where $N$ grows with $T$ is a natural complement. At the intersection of the two regimes, a composite likelihood that
exploits both the individual binary structure and the aggregate Poisson approximation may yield efficiency gains relative to either likelihood alone. On the empirical side, the binary panel framework is well suited to modeling correlated binary decisions in large administrative datasets: voting records, loan defaults, and insurance claims. The network GAB model in particular provides a tractable vehicle for studying the role of connectivity in propagating financial or economic shocks.

\appendix

\section{Proofs: Existence and Uniqueness}

\begin{proof}[Proof of Theorem \ref{th_stationary_sol}]
We use \citet[Theorem 12.0.1 (Markov–Krylov–Bogolyubov theorem)]{meyn2009markov}, which guarantees the existence of a stationary solution. Since the Markov chain $(\p_{t-1},\ldots,\p_{t-s},\y_{t-1},\ldots,\y_{t-q})$ is defined on a compact space $X=[0,1]^{Ns}\times\{0,1\}^{Nq}$, the ``bounded in probability on average'' condition is satisfied. It only remains to check that the chain is weak Feller.

Let $C\left(X\right)$ be the class of bounded continuous functions from $X$ to $\mathbb{R}$.  To show that the chain is weak Feller we need to check that its transition probability kernel maps $C\left(X\right)$ to $C\left(X\right)$.

Take any bounded and continuous $f:X\to\mathbb{R}$. The transition kernel maps
\begin{equation*}\begin{split}
  &(Pf)(\p_{t-1},\ldots,\p_{t-s},\y_{t-1},\ldots,\y_{t-q})\\
  &= \sum\limits_{\y\in\{0,1\}^N} \prod_{i=1}^N p_i^{y_i}(1-p_i)^{1-y_i}
  f\left(\p,\p_{t-1},\ldots,\p_{t-s+1},  \y,\y_{t-1},\ldots,\y_{t-q+1}\right) ,\\
  &\p=(p_1,\ldots,p_N)=\g\left(\{\p_{\tau}\}_{\tau=t-1,\ldots,t-s},\,\{\y_{\tau}\}_{\tau=t-1,\ldots,t-q}\right),\quad \y=(y_1,\ldots,y_N).
\end{split}\end{equation*}
Since both $f$ and $\g$ are continuous, their composition in $(Pf)$ is also continuous. Function $\g$ is bounded since its image is $[0,1]^N$, and so the composition $(Pf)$ is also bounded. Thus, we are left with a weak Feller chain on a compact set, which by \citet[Theorem 12.0.1 (Markov–Krylov–Bogolyubov theorem)]{meyn2009markov}, has an invariant probability measure (i.e., stationary solution).
\end{proof}

\begin{proof}[Proof of Theorem \ref{th_stationary_unique}]
Let $(\y_t,\p_t)_{t\in\mathbb Z}$ be a stationary solution given by Theorem \ref{th_stationary_sol} and let $\F_t$ be the corresponding $\sigma$-algebras. Extending the $\sigma$-algebras, if necessary, we assume that they support i.i.d.\ uniform random variables $u_{i,t}$ from \eqref{eq_dgp_y}.

Let $(\y'_t,\p'_t)_{t\ge t_0}$ be (possibly another) solution of \eqref{eq_dgp_y}--\eqref{eq_dgp_p} constructed by the same sequence $u_{i,t}$, but potentially different starting values
$\p':=\{\p'_{t_0+s-1},\ldots,\p'_{t_0}\},\,\y':=\{\y'_{t_0+q-1},\ldots,\y'_{t_0}\}$; we need another extension of $\sigma$-algebras to support these starting values. Observe that
\begin{equation}\label{eq_Ey_Ep}
\begin{split}
\mathbb{E}_{t-1}|y_{i,t}-y'_{i,t}|&=\mathbb{E}_{t-1}\left|\1(u_{i,t}\leq p_{i,t})-\1(u_{i,t}\leq p'_{i,t})\right|\\
&=\mathbb{E}_{t-1}\1(p'_{i,t}<u_{i,t}\leq p_{i,t} \text{ or }p_{i,t}<u_{i,t}\leq p'_{i,t})\\
&=\mathbb{E}_{t-1}\1(\min(p_{i,t},p'_{i,t})<u_{i,t}\leq\max(p_{i,t},p'_{i,t}))\\
&=\max(p_{i,t},p'_{i,t})-\min(p_{i,t},p'_{i,t})=|p_{i,t}-p'_{i,t}|.
\end{split}\end{equation}
Thus, assuming Assumption \ref{ass_contraction1} holds,
\begin{equation}\label{eq_AR_E}
\begin{split}
\mathbb{E}|p_{i,t}-p'_{i,t}|&=
\mathbb{E}\left|g_i\left(\{\p_{\tau}\}_{\tau=t-1,\ldots,t-s},\,\{\y_{\tau}\}_{\tau=t-1,\ldots,t-q}\right)
-g_i\left(\{\p'_{\tau}\}_{\tau=t-1,\ldots,t-s},\,\{\y'_{\tau}\}_{\tau=t-1,\ldots,t-q}\right)\right|\\
&\leq\mathbb{E}\left(\sum\limits_{j=1}^{N}\sum\limits_{\tau=1}^{s} |p_{j,t-\tau}-p_{j,t-\tau}'|\beta_{ij}^{\tau}
 +\sum\limits_{j=1}^{N}\sum\limits_{\tau=1}^{q} |y_{j,t-\tau}-y_{j,t-\tau}'|\alpha_{ij}^{\tau}\right)\\
 &=\sum\limits_{j=1}^{N}\sum\limits_{\tau=1}^{s} \mathbb{E}|p_{j,t-\tau}-p_{j,t-\tau}'|\beta_{ij}^{\tau}
 +\sum\limits_{j=1}^{N}\sum\limits_{\tau=1}^{q} \mathbb{E}|p_{j,t-\tau}-p_{j,t-\tau}'|\alpha_{ij}^{\tau}\\
 &=\sum\limits_{j=1}^{N}\sum\limits_{\tau=1}^{\max(s,q)} \mathbb{E}|p_{j,t-\tau}-p_{j,t-\tau}'|\left(\alpha_{ij}^{\tau}+\beta_{ij}^{\tau}\right),
\end{split}\end{equation}
where we treat $\alpha_{ij}^{\tau}=0$ for $\tau>q$ and $\beta_{ij}^{\tau}=0$ for $\tau>s$.

Consider an $N\max(s,q)$-dimensional vector
\begin{equation*}\begin{split}
&\eta_t=\left(\mathbb{E}|p_{1,t}-p'_{1,t}|,\ldots,\mathbb{E}|p_{N,t}-p'_{N,t}|,\ldots,\mathbb{E}|p_{N,t-\max(s,q)+1}-p'_{N,t-\max(s,q)+1}|\right)^{\T}.
\end{split}\end{equation*}
Then Eq.~\eqref{eq_AR_E} implies
\begin{equation}\label{eq_stat_geom_bound}
\eta_t\leq\Phi\eta_{t-1}\leq\Phi^{t-t_0-\max(s,q)}\eta_{t_0+\max(s,q)}.
\end{equation}
By Assumption \ref{ass_contraction1}, all eigenvalues of $\Phi$ lie inside the unit circle. Thus, $\Phi^t\to0$ as $t\to\infty$, so that $\mathbb{E}|p_{i,t}-p'_{i,t}|\to0$ geometrically fast for all $i$ as $t\to\infty$. Thus, we must have $\p_t-\p'_t\xrightarrow[t\to\infty]{p}0$.



Similarly, if instead we assume that Assumption \ref{ass_contraction2} holds,
\begin{equation}\label{eq_AR_Esum}
\begin{split}
\mathbb{E}\sum\limits_{i=1}^{N}|p_{i,t}-p'_{i,t}|&=
\mathbb{E}\sum\limits_{i=1}^{N}\left|g_i\left(\{\p_{\tau}\}_{\tau=t-1,\ldots,t-s},\,\{\y_{\tau}\}_{\tau=t-1,\ldots,t-q}\right)
-g_i\left(\{\p'_{\tau}\}_{\tau=t-1,\ldots,t-s},\,\{\y'_{\tau}\}_{\tau=t-1,\ldots,t-q}\right)\right|\\
&\leq K\mathbb{E}\left(\sum\limits_{i=1}^{N}\sum\limits_{\tau=1}^{s} |p_{i,t-\tau}-p_{i,t-\tau}'|
 +\sum\limits_{i=1}^{N}\sum\limits_{\tau=1}^{q} |y_{i,t-\tau}-y_{i,t-\tau}'|\right)\\
 &=K\sum\limits_{\tau=1}^{s}\mathbb{E}\sum\limits_{i=1}^{N}|p_{i,t-\tau}-p'_{i,t-\tau}|+K\sum\limits_{\tau=1}^{q}\mathbb{E}\sum\limits_{i=1}^{N}|p_{i,t-\tau}-p'_{i,t-\tau}|.
\end{split}\end{equation}

Consider a $\max(s,q)$-dimensional vector
$$
\tilde\eta_t=\left(\mathbb{E}\sum\limits_{i=1}^{N}|p_{i,t}-p'_{i,t}|,\ldots,\mathbb{E}\sum\limits_{i=1}^{N}|p_{i,t-\max(s,q)+1}-p'_{i,t-\max(s,q)+1}|\right)
$$
and a $\max(s,q)\times \max(s,q)$ matrix $\tilde\Phi$, where
$$\tilde\Phi_{1i}=2K,\,i=1,\ldots,\min(s,q),\qquad \tilde\Phi_{1i}=K,\,i=\min(s,q)+1,\ldots,\max(s,q),$$
the bottom-left $(\max(s,q)-1)\times (\max(s,q)-1)$ block of $\tilde\Phi$ is an identity matrix, and the bottom-right $(\max(s,q)-1)\times 1$ block is zero. Then the $L_1$ norm of the first row of $\tilde\Phi$ equals $K(s+q)<1$ and, again, all eigenvalues of $\tilde\Phi$ lie inside the unit circle. Thus,
$$\tilde\eta_t\leq\tilde\Phi\tilde\eta_{t-1}\leq\tilde\Phi^{t-t_0-\max(s,q)}\tilde\eta_{t_0+\max(s,q)}\to0\quad\text{ as }t\to\infty$$
and again $\mathbb{E}|p_{i,t}-p'_{i,t}|\to0$ geometrically fast and $\p_t-\p'_t\xrightarrow[t\to\infty]{p}0$.

Geometric decay of $\mathbb{E}|p_{i,t}-p'_{i,t}|$ and boundedness of $p_{i,t}, p'_{i,t}$ imply \eqref{eq_geometeric_mixing}. The convergence in probability $\p_t-\p'_t\xrightarrow[t\to\infty]{p}0$ implies that the L\'{e}vy-Prokhorov distance between the distributions of $(\y_{t-1},\dots, \y_{t-q},\, \p_{t-1},\dots, \p_{t-s})$ and $(\y'_{t-1},\dots, \y'_{t-q},\, \p'_{t-1},\dots, \p'_{t-s})$ goes to $0$ as $t\to \infty$.
Hence, if $(\y'_t,\p'_t)$ is another stationary solution, then, by  $t$--independence of the distribution, we should have the distributional match $(\y_{t-1},\dots, \y_{t-q},\, \p_{t-1},\dots, \p_{t-s})\stackrel{d}=(\y'_{t-1},\dots, \y'_{t-q},\, \p'_{t-1},\dots, \p'_{t-s})$ for all $t$. Using the Markov chain representation from the proof of Theorem \ref{th_stationary_sol},
this implies that the distributions of any two stationary solutions coincide.\
\end{proof}

The following corollary will be useful for estimation theorems.
\begin{corollary}\label{cor_lln}
Suppose that $\g(\cdot)$ satisfies either Assumption \ref{ass_contraction1} or \ref{ass_contraction2}, and let $\{(\y_t,\p_t)\}_{t\in\mathbb Z}$ denote
the unique strictly stationary solution of \eqref{eq_dgp_y}--\eqref{eq_dgp_p}. Define the augmented Markov state
\[
    Z_t
    :=
    (\p_t,\ldots,\p_{t-s+1},
     \y_t,\ldots,\y_{t-q+1})
\]
and let $\pi$ denote its stationary distribution. Then, for every
continuous function $h$ on the state space of $Z_t$,
\begin{equation}\label{eq:lln_prob}
    \frac{1}{T}\sum_{t=1}^T h(Z_t)\xrightarrow[T\to\infty]{p}\mathbb{E}_\pi[h(Z_0)] .
\end{equation}
\end{corollary}
\begin{proof}
Let $P$ denote the transition kernel of the Markov chain $\{Z_t\}$. By Theorem \ref{th_stationary_unique}, for every initial state $z$,
\begin{equation}\label{eq:Pnweak}
    P^n(z,\cdot)\xrightarrow[n\to\infty]{}\pi.
\end{equation}
Since the state space is compact, every continuous function $h$ is bounded. Hence Eq.~\eqref{eq:Pnweak} implies
\[
    P^n h(z)
    :=
    \int h(z')P^n(z,dz')
    \longrightarrow
    \int h(z')\pi(dz')
    =:\pi(h)
\]
for every $z$.

Since the process is stationary,
\[
\begin{split}
    Cov_\pi(h(Z_0),h(Z_n))
    &=\int h(z)P^nh(z)\,\pi(dz)-\pi(h)^2 .
\end{split}
\]
Moreover,
\[
    |h(z)P^nh(z)|\leq \|h\|_\infty^2 .
\]
Therefore, by dominated convergence,
\begin{equation}\label{eq:cov_decay}
    Cov_\pi(h(Z_0),h(Z_n))\to0.
\end{equation}

Let
\[
    \gamma_n:=Cov_\pi(h(Z_0),h(Z_n)).
\]
By stationarity,
\[
    \mathbb{V}\left(
        \frac{1}{T}\sum_{t=1}^T h(Z_t)
    \right)
    =
    \frac{\gamma_0}{T}
    +\frac{2}{T^2}
      \sum_{n=1}^{T-1}(T-n)\gamma_n
    =
    \frac{\gamma_0}{T}
    +\frac{2}{T}
      \sum_{n=1}^{T-1}
      \left(1-\frac{n}{T}\right)\gamma_n\to0.
\]
Since, by stationarity,
\[
    \mathbb{E}\left[
        \frac{1}{T}\sum_{t=1}^T h(Z_t)
    \right]
    =\pi(h),
\]
we conclude that
\[
    \frac{1}{T}\sum_{t=1}^T h(Z_t)
    \xrightarrow[T\to\infty]{p} \pi(h).\qedhere
\]
\end{proof}

\section{Proofs: Aggregation}\label{append_aggreg}

\begin{lemma}\label{th_Poisson_randon}
Consider a sequence of random probabilities $q_i, i=1,\ldots,N$, such that
\begin{equation}\label{eq_Poiss_cond_lemma}
  \sum_{i=1}^N q_{i}\xrightarrow[N\to\infty]{d}\xi,\qquad \max\limits_{i=1,\ldots,N} q_{i}\xrightarrow[N\to\infty]{p}0,
\end{equation}
where $\xi$ is some random variable. Given $q_i,i=1,\ldots,N$, consider a sequence of independent Bernoulli random variables ${y_i\thicksim\, B(q_i)}$. As $N\to\infty$, $\left[\sum_{i=1}^N q_i,\,\sum_{i=1}^N y_i\right]$ converges in distribution to $[\xi,Y]$, where the random variable $Y$ is conditionally Poisson distributed with random intensity $\xi$.
\end{lemma}
\begin{proof}
Letting $S_{N}:=\sum\limits_{i=1}^{N}y_{i}$, the conditional (on the realization $\{q_i\}_i$) probability generating function of $S_N$ satisfies
\begin{equation}\label{eq_G_N(z)}
  G_{N}(z):=\mathbb{E} (z^{S_{N}}\mid q_1,\ldots,q_N)=\prod\limits_{i=1}^N\left(q_{i}z+(1-q_{i})\right)
  =\prod\limits_{i=1}^N\left(1+q_{i}(z-1)\right)
\end{equation}
so that
\begin{equation*}\begin{split}
  \log(G_{N}(z))
  &=\sum\limits_{i=1}^N \log\left(1+q_{i}(z-1)\right)
  =\sum\limits_{i=1}^N\left(q_{i}(z-1)-\frac1{2}q_{i}^2(z-1)^2+o(q_{i}^2)\right)\\
  &=(z-1)\sum\limits_{i=1}^N q_{i}-\frac1{2}(z-1)^2\sum\limits_{i=1}^N q_{i}^2+o\left(\sum\limits_{i=1}^N q_{i}^2\right).
\end{split}\end{equation*}
Since $\sum\limits_{i=1}^N q_{i}^2\leq\left(\max_i q_i\right)\sum\limits_{i=1}^N q_{i}\xrightarrow[N\to\infty]{p}0$, as $N\to\infty$, $\log(G_{N}(z))\xrightarrow[N\to\infty]{d} (z-1)\xi$, so that $G_{N}(z)$ converges in distribution to the generating function of the Poisson distribution, $\exp((z-1)\xi)$. Thus, $S_{N}\xrightarrow[N\to\infty]{d} Y$ with $Y \sim Pois(\xi)$ jointly with $\sum_{i=1}^N q_{i}\xrightarrow[N\to\infty]{d}\xi$.
\end{proof}

\begin{proof}[Proof of Theorem \ref{th_aggregation}]
We show by induction that, for all $T=0,1,\ldots$,
$\bigl[\sum_{i=1}^{N}p_{i,t},\,\sum_{i=1}^{N}y_{i,t}\bigr]_{t=0,\ldots,T}$ converges both in distribution and in expectation to a Poisson autoregression $\bigl[\lambda_t,X_t\bigr]_{t=0,\ldots,T}$ and, additionally, $\max_i p_{i,T}\xrightarrow[N\to\infty]{p}0$.

At time $t=0$ (base of induction), the process satisfies conditions of Lemma \ref{th_Poisson_randon}, thus,
\begin{equation}\label{eq_agg_t0}
\left[\sum\limits_{i=1}^N p_{i,0},\,\sum\limits_{i=1}^N y_{i,0}\right]\xrightarrow[N\to\infty]{d} [\lambda_0,X_0],~{\rm with}~X_0
\sim Pois(\lambda_0).
\end{equation}
By \eqref{eq_aggregation_pi0}, $\max_i p_{i,0}\xrightarrow[N\to\infty]{p}0$ and $\mathbb{E}\sum\limits_{i=1}^N p_{i,0}\xrightarrow[N\to\infty]{}\mathbb{E}\lambda_0$. Applying conditional expectation,
\begin{equation*}
\begin{split}
\mathbb{E}\sum\limits_{i=1}^N y_{i,0}&=\mathbb{E}\sum\limits_{i=1}^N p_{i,0}\xrightarrow[N\to\infty]{}\mathbb{E}\lambda_0=\mathbb{E}\bigl[\mathbb{E}(X_0\mid\lambda_0)\bigr]=\mathbb{E}X_0.
\end{split}
\end{equation*}

Let us show that $T=1$ follows from $T=0$. The argument for the general induction step is the same.

Let us add $\sum\limits_{i=1}^N p_{i,1}$ to the joint convergence. By Markov inequality for any $\eps>0$,
\begin{equation}\label{eq_Markov_bound}
\rm{Prob}\left(\frac1{N^{\kappa}}\sum\limits_{j=1}^{N} y_{j,0}>\eps\right)\leq\frac1{\eps}\mathbb{E}\frac1{N^{\kappa}}\sum\limits_{j=1}^{N} y_{j,0}
=\frac1{N^{\kappa}\eps}\sum\limits_{j=1}^{N} \mathbb{E}p_{j,0}\xrightarrow[N\to\infty]{}0.
\end{equation}
Thus, summing Eq.~\eqref{eq_multiple_k_diff} with respect to $i$ for $t=1$,
$$
\sum\limits_{i=1}^N p_{i,1}
=\frac{1}{N}\sum\limits_{i=1}^N c_{i}+\frac{1}{N^{\kappa}}\sum\limits_{i=1}^N a_{i}y_{i,0}
+\beta\sum\limits_{i=1}^N p_{i,0}+\frac1{N}\sum\limits_{j=1}^{N}y_{j,0} \sum\limits_{i=1}^N \gamma_{i}
\xrightarrow[N\to\infty]{d} \bar{c}+\beta\lambda_0+\bar{\gamma}X_0=\lambda_1,
$$
and we get
$$\left[\sum\limits_{i=1}^N p_{i,0},\,\sum\limits_{i=1}^N y_{i,0},\,\sum\limits_{i=1}^N p_{i,1},\right]\xrightarrow[N\to\infty]{d} [\lambda_0,X_0,\lambda_1].$$

Let us show that the conditions \eqref{eq_Poiss_cond_lemma} are valid for $q_i=p_{i,1}$. The first one is already proved. To show the second one, note that by Eq.~\eqref{eq_multiple_k_diff},
$$
p_{i,1}=\frac{c_i}{N}+\frac{a_i}{N^{\kappa}}y_{i,0}+\beta p_{i,0}+\gamma_i\frac1{N}\sum\limits_{j=1}^{N} y_{j,0}.
$$
Thus, since by assumptions, $\max_i p_{i,0}\to0$ and $0\leq\gamma_i\leq1$, and, by \eqref{eq_Markov_bound} with $\kappa=1$, $\frac1{N}\sum\limits_{j=1}^{N} y_{j,0}\xrightarrow[N\to\infty]{p}0$,
$$0\leq\max_i p_{i,1}\leq\frac{C}{N}+\frac{C}{N^{\kappa}}+\beta\max_i p_{i,0}+\frac1{N}\sum\limits_{j=1}^{N} y_{j,0}\xrightarrow[N\to\infty]{}0.$$
Direct application of Lemma \ref{th_Poisson_randon} gives
\begin{equation}\label{eq_agg_t1}
\left[\sum\limits_{i=1}^N p_{i,1},\,\sum\limits_{i=1}^N y_{i,1}\right]\xrightarrow[N\to\infty]{d} [\lambda_1,X_1],~{\rm with}~X_1\sim
Pois(\lambda_1).
\end{equation}
We need to combine the above with distributional convergence \eqref{eq_agg_t0}, which requires an extension of Lemma \ref{th_Poisson_randon}. By adding $\{p_{i,0},y_{i,0}\}_{i=1}^{N}$ to the conditioning in \eqref{eq_G_N(z)}, and repeating the same proof, the convergence in \eqref{eq_agg_t1} extends to conditional convergence. Hence, we arrive at
$$
\left[\sum\limits_{i=1}^N p_{i,0},\,\sum\limits_{i=1}^N y_{i,0},\,\sum\limits_{i=1}^N p_{i,1},\,\sum\limits_{i=1}^N y_{i,1}\right]\xrightarrow[N\to\infty]{d} [\lambda_0,X_0,\lambda_1,X_1].
$$
It remains to prove convergence in expectations. For $t=0$ we have already shown it. For $t=1$,
\begin{equation*}\begin{split}
\sum\limits_{i=1}^N \mathbb{E}p_{i,1}&
=\frac{1}{N}\sum\limits_{i=1}^N c_{i}+\frac{1}{N^{\kappa}}\sum\limits_{i=1}^N a_{i}\mathbb{E}y_{i,0}
+\beta\sum\limits_{i=1}^N \mathbb{E}p_{i,0}+\frac1{N}\sum\limits_{j=1}^{N}\mathbb{E}y_{j,0} \sum\limits_{i=1}^N \gamma_{i}\\
&=\frac{1}{N}\sum\limits_{i=1}^N c_{i}+\frac{1}{N^{\kappa}}\sum\limits_{i=1}^N a_{i}\mathbb{E}p_{i,0}
+\beta\sum\limits_{i=1}^N \mathbb{E}p_{i,0}+\frac1{N}\sum\limits_{j=1}^{N}\mathbb{E}y_{j,0} \sum\limits_{i=1}^N \gamma_{i}\\
&\xrightarrow[N\to\infty]{}\bar{c}+\beta\mathbb{E}\lambda_0+\bar{\gamma}\mathbb{E}X_0=\mathbb{E}\lambda_1,\\
\sum\limits_{i=1}^N \mathbb{E}y_{i,1}&
=\sum\limits_{i=1}^N \mathbb{E}p_{i,1}\xrightarrow[N\to\infty]{}\mathbb{E}\lambda_1=\mathbb{E}\bigl[\mathbb{E}(X_1\mid\lambda_1)\bigr]=\mathbb{E}X_1.
\end{split}\end{equation*}
Repeating the same steps for general $T$, we arrive at the desired joint convergence.
\end{proof}

\begin{proof}[Proof of Proposition \ref{rem_init_cond_aggreg}]
Consider a stationary distribution and let us denote it as $\p^{\ast}=(p^{\ast}_1,\ldots,p^{\ast}_N)^{\T}$ with corresponding realizations $\y^{\ast}=(y^{\ast}_1,\ldots,y^{\ast}_N)^{\T}$. Let us start with the mean of $\sum\limits_{i=1}^{N}p^{\ast}_{i,t}$. Plugging $\beta_i=\beta$, $\omega_i=\tfrac{c_i}{N}$, $\alpha_i=\tfrac{a_i}{N^{\kappa}}$ in Eq.~\eqref{eq_multiple_k_diff_mean_sum}, we get
$$\sum\limits_{i=1}^{N}\mu_i=\frac{\frac1{N}\sum\limits_{i=1}^{N}\frac{c_i}{1-\beta-a_i/N^{\kappa}}}{1-\frac1{N}\sum\limits_{i=1}^{N}\frac{\gamma_i}{1-\beta-a_i/N^{\kappa}}}.$$
Taylor expanding, we get
$$\gamma_i(1-\beta-a_i/N^{\kappa})^{-1}\approx\gamma_i(1-\beta)^{-1}\left(1+\frac{a_i}{(1-\beta)N^{\kappa}}\right),$$
$$c_i(1-\beta-a_i/N^{\kappa})^{-1}\approx c_i(1-\beta)^{-1}\left(1+\frac{a_i}{(1-\beta)N^{\kappa}}\right)$$
and
\begin{equation}\label{eq_sum_mu_i}
\sum\limits_{i=1}^{N}\mu_i=\frac{\frac1{N}\sum\limits_{i=1}^{N}c_i}{1-\beta-\frac1{N}\sum\limits_{i=1}^{N}\gamma_i}+O(N^{-\kappa})
\xrightarrow[N\to\infty]{}\frac{\bar{c}}{1-\beta-\bar{\gamma}}.
\end{equation}
Similarly, Eq.~\eqref{eq_multiple_k_diff_mean} implies, uniformly in $i$,
\begin{equation}\label{eq_mu_i_order}
\mu_i
=\frac{c_i}{N(1-\beta-a_iN^{-\kappa})}+\frac{\gamma_i}{N(1-\beta-a_iN^{-\kappa})}\sum_{j=1}^N \mu_{j}=O(N^{-1}).
\end{equation}
We begin by checking that $\max\limits_{i=1,\ldots,N} p^{\ast}_{i,t}\xrightarrow[N\to\infty]{p}0$. First, note that
\begin{equation*}\begin{split}
  0\leq\sum_{i} p_{i,t}^{\ast}&=\frac1{N}\sum_i c_i+\frac1{N^{\kappa}}\sum_i a_i y_{i,t-1}^{\ast}+\beta\sum_{i} p_{i,t-1}^{\ast}+\frac1{N}\sum_i\gamma_i\sum_i y_{i,t-1}^{\ast}\\
  &\leq C+\beta\sum_{i} p_{i,t-1}^{\ast}+\left(\frac1{N}\sum_i\gamma_i+\frac{C}{N^{\kappa}}\right)\sum_i y_{i,t-1}^{\ast}.
\end{split}\end{equation*}
Thus, by Minkowski's inequality (or Cauchy–Schwarz inequality), letting $\tilde{\gamma}_N:=\frac1{N}\sum_i\gamma_i+\frac{C}{N^{\kappa}}$,
\begin{equation}\label{eq_sum_i_p_i_sq}
\begin{split}
  \sqrt{\mathbb{E}\left(\sum_{i} p_{i,t}^*\right)^2}
  &\leq C+\sqrt{\mathbb{E}\left(\beta\sum_{i} p_{i,t-1}^{\ast}+\tilde{\gamma}_N\sum_i y_{i,t-1}^{\ast}\right)^2}\\
  &=C+\sqrt{\mathbb{E}\left(\left(\beta+\tilde{\gamma}_N\right)\sum_{i} p_{i,t-1}^{\ast}+\tilde{\gamma}_N\sum_i \left(y_{i,t-1}^{\ast}-p_{i,t-1}^{\ast}\right)\right)^2}\\
  &=C+\sqrt{\left(\beta+\tilde{\gamma}_N\right)^2\mathbb{E}\left(\sum_{i}p_{i,t-1}^{\ast}\right)^2
  +\tilde{\gamma}_N^2\mathbb{E}\left(\sum_{i} p_{i,t-1}^{\ast}-\sum_{i} (p_{i,t-1}^{\ast})^2\right)}\\
  &\leq C+\left(\beta+\tilde{\gamma}_N\right)\sqrt{\mathbb{E}\left(\sum_{i} p_{i,t}^*\right)^2}+\tilde{\gamma}_N\sqrt{\sum_i\mu_i}
\end{split}\end{equation}
where the third line follows from $\mathbb{E}_{t-2}(y_{i,t-1}^{\ast}-p_{i,t-1}^{\ast})=\mathbb{E}_{t-2}(y_{i,t-1}^{\ast}-p_{i,t-1}^{\ast})(y_{j,t-1}^{\ast}-p_{j,t-1}^{\ast})=0$ for any $i\neq j$, and the last line uses stationarity. Since $\beta+\gamma_i<\bar{C}<1$ for all i, $1-\beta-\tilde{\gamma}_N>1-\bar{C}-\frac{C}{N^{\kappa}}>0$ for all $N$ large enough. Thus, Eq.~\eqref{eq_sum_i_p_i_sq} combined with the bound \eqref{eq_sum_mu_i}, imply that
\begin{equation}\label{eq_p_i_tight}
  \sup_{N}\mathbb{E}\left(\sum_{i} p_{i,t}^*\right)^2<\infty.
\end{equation}
Second, again using $\mathbb{E}_{t-1}(y_{i,t}^{\ast}-p_{i,t}^{\ast})=\mathbb{E}_{t-1}(y_{i,t}^{\ast}-p_{i,t}^{\ast})(y_{j,t}^{\ast}-p_{j,t}^{\ast})=0$,
\begin{equation}\label{eq_sum_i_y_i_sq}
\begin{split}
\mathbb{E}\left(\sum_{i} y_{i,t}^*\right)^2&=\mathbb{E}\left(\sum_{i} (y_{i,t}^*-p_{i,t}^*)+\sum_{i} p_{i,t}^*\right)^2
=\mathbb{E}\left(\sum_{i} (y_{i,t}^*-p_{i,t}^*)\right)^2+\mathbb{E}\left(\sum_{i} p_{i,t}^*\right)^2\\
&=\mathbb{E}\sum_{i}\left(p_{i,t}^*-(p_{i,t}^*)^2\right)+\mathbb{E}\left(\sum_{i} p_{i,t}^*\right)^2\leq\sum_i\mu_i+\mathbb{E}\left(\sum_{i} p_{i,t}^*\right)^2=O(1).
\end{split}\end{equation}
Third, again using stationarity and Minkowski's inequality for the dynamic of $p_{i,t}^*$,
\begin{equation*}
\sqrt{\mathbb{E}\left(p_{i,t}^*\right)^2}
\leq\frac{C}{N}+\frac{C}{N^{\kappa}}\sqrt{\mu_i}+\beta\sqrt{\mathbb{E}\left(p_{i,t}^*\right)^2}
+\frac{1}{N}\sqrt{\mathbb{E}\left(\sum_{i} y_{i,t}^*\right)^2}.
\end{equation*}
Combining the above bound with Eq.~\eqref{eq_mu_i_order}, \eqref{eq_sum_i_y_i_sq}, we get, uniformly in $i$,
\begin{equation*}
(1-\beta)\sqrt{\mathbb{E}\left(p_{i,t}^*\right)^2}
\leq \frac{C}{N}+\frac{C}{N^{\kappa}}\sqrt{\mu_i}+\frac{1}{N}\sqrt{\mathbb{E}\left(\sum_{i} y_{i,t}^*\right)^2}=O(N^{-1}+N^{-\kappa-1/2})
\end{equation*}
and
\begin{equation*}
\mathbb{E}\left(p_{i,t}^*\right)^2=O(N^{-2}+N^{-2\kappa-1}),\qquad \sum_{i=1}^{N}\mathbb{E}\left(p_{i,t}^*\right)^2=O(N^{-1}+N^{-\kappa}).
\end{equation*}
Thus, for any $\eps>0$,
\begin{equation*}\begin{split}
Prob\left(\max\limits_{i=1,\ldots,N} p^{\ast}_{i,t}>\eps\right)&\leq\sum\limits_{i=1}^{N}Prob(p^{\ast}_{i,t}>\eps)=\sum\limits_{i=1}^{N}Prob\left((p^{\ast}_{i,t})^2>\eps^2\right)\\
&\leq\frac1{\eps^2} \sum\limits_{i=1}^{N}\mathbb{E}(p^{\ast}_{i,t})^2=\frac1{\eps^2} O(N^{-1}+N^{-\kappa})\xrightarrow[N\to\infty]{}0.
\end{split}\end{equation*}

It remains to check convergence of $\sum_i p^{\ast}_{i,t}$. Consider its probability generating function, where $\bar{\gamma}_N:=\frac1{N}\sum_i\gamma_i$,
\begin{equation}\label{eq_G_N(z)_sum_pi}
\begin{split}
  G_{N}(z)&:=\mathbb{E} z^{\sum_i p^{\ast}_{i,t}}=
  \mathbb{E} z^{\frac1{N}\sum_i c_i+\sum_i \left(\frac{a_i}{N^{\kappa}}+\bar{\gamma}_N\right) y_{i,t-1}^{\ast}+\beta\sum_{i} p_{i,t-1}^{\ast}}\\
  &=\mathbb{E}\left\{ z^{\frac1{N}\sum_i c_i+\beta\sum_{i} p_{i,t-1}^{\ast}}
  \prod_{i=1}^{N}\left[1-p_{i,t-1}^{\ast}\left(1-z^{\frac{a_i}{N^{\kappa}}+\bar{\gamma}_N}\right)\right]\right\}.
\end{split}\end{equation}
The product in \eqref{eq_G_N(z)_sum_pi} can be analyzed by means of
\begin{equation*}\begin{split}
&\log\prod_{i=1}^{N}\left[1-p_{i,t-1}^{\ast}\left(1-z^{\frac{a_i}{N^{\kappa}}+\bar{\gamma}_N}\right)\right]=
\sum\limits_i\log\left[1-p_{i,t-1}^{\ast}\left(1-z^{\frac{a_i}{N^{\kappa}}+\bar{\gamma}_N}\right)\right]\\
&=-\sum\limits_ip_{i,t-1}^{\ast}\left(1-z^{\frac{a_i}{N^{\kappa}}+\bar{\gamma}_N}\right)
+o\left(\sum\limits_ip_{i,t-1}^{\ast}\left(1-z^{\frac{a_i}{N^{\kappa}}+\bar{\gamma}_N}\right)\right)\\
&=-\sum\limits_ip_{i,t-1}^{\ast}\left(1-z^{\bar{\gamma}_N}\right)
+o\left(\sum\limits_ip_{i,t-1}^{\ast}\right)+o\left(\sum\limits_ip_{i,t-1}^{\ast}\left(1-z^{\frac{a_i}{N^{\kappa}}+\bar{\gamma}_N}\right)\right).
\end{split}\end{equation*}
For each fixed $z\in(0,1]$,
$$0\leq \sum\limits_ip_{i,t-1}^{\ast}\left(1-z^{\frac{a_i}{N^{\kappa}}+\bar{\gamma}_N}\right)\leq\tilde{C}\sum\limits_ip_{i,t-1}^{\ast}=O_p(1),$$
where the boundedness in probability follows from \eqref{eq_p_i_tight} by Markov inequality. Therefore, for each fixed $z\in(0,1]$ we can rewrite \eqref{eq_G_N(z)_sum_pi} as
\begin{equation}\label{eq_G_N_new}
\begin{split}
G_{N}(z)&=\mathbb{E} z^{\frac{\sum_i c_i}{N}+\beta\sum_{i} p_{i,t-1}^{\ast}}
\exp\left[-\sum\limits_ip_{i,t-1}^{\ast}\left(1-z^{\bar{\gamma}_N}\right)+o(1)\right]\\
&=z^{\frac{\sum_i c_i}{N}}\mathbb{E}\left[z^{\beta}\exp\{z^{\bar{\gamma}_N}-1\}\right]^{\sum\limits_ip_{i,t-1}^{\ast}}(1+o(1))
=z^{\frac{\sum_i c_i}{N}}G_{N}\left(z^{\beta}\exp\{z^{\bar{\gamma}_N}-1\}\right)+o(1).
\end{split}\end{equation}
Uniform integrability \eqref{eq_p_i_tight} implies tightness and, thus, that $\sum_i p^{\ast}_{i,t}$ has a convergent subsequence. Denote the limit of this subsequence by $\lambda_0$ and let $G_z=\mathbb{E}z^{\lambda_0}$. For this subsequence we must have, for each $z\in(0,1]$, $G_{N}(z)\xrightarrow[N\to\infty]{}G(z)$. Moreover, $z^{\frac{\sum_i c_i}{N}}\xrightarrow[N\to\infty]{}z^{\bar{c}}$ and $z^{\beta}\exp\{z^{\bar{\gamma}_N}-1\}\xrightarrow[N\to\infty]{}z^{\beta}\exp\{z^{\bar{\gamma}}-1\}$. Thus, taking the limit in \eqref{eq_G_N_new} and using the continuity and monotonicity of generating functions, we get for any fixed $z\in(0,1)$
\begin{equation}\label{eq_G_N_limit}
  G(z)=z^{\bar{c}}G\left(z^{\beta}\exp\{z^{\bar{\gamma}}-1\}\right).
\end{equation}
We claim that there is a unique solution to \eqref{eq_G_N_limit}. Fix $z_0\in(0,1)$ and let $z_{n+1}:=z_n^{\beta}\exp\{z_n^{\bar{\gamma}}-1\}$ for $n=0,1,\ldots$. Then $z_n\in(0,1)$ and, using the bound $1-\exp{-z}\leq z$ for $z\geq0$,
\begin{equation}\label{eq_log z_n_bound}
\begin{split}
0\leq-\log z_{n+1}&=-\beta\log z_n+1-z_n^{\bar{\gamma}}=-\beta\log z_n+1-\exp\{\bar{\gamma}\log z_n\}\\
&\leq-(\beta+\bar{\gamma})\log z_n\leq-(\beta+\bar{\gamma})^n\log z_0\xrightarrow[n\to\infty]{}0.
\end{split}\end{equation}
Thus, $z_{n+1}\xrightarrow[n\to\infty]{}1$, so that $G(z_n)\xrightarrow[n\to\infty]{}G(1)\equiv1$. Iterating \eqref{eq_G_N_limit}, we get
\begin{equation*}\begin{split}
G(z_0)=z_0^{\bar{c}}G\left(z_1\right)=z_0^{\bar{c}}z_1^{\bar{c}}G\left(z_2\right)
=\prod\limits_{i=0}^{\infty} z_i^{\bar{c}}=\exp\left\{\bar{c}\sum\limits_{i=0}^{\infty} \log z_i\right\},
\end{split}\end{equation*}
and the sum $\sum\limits_{i=0}^{\infty} \log z_i$ is finite, since by \eqref{eq_log z_n_bound} it is bounded by the sum of a geometric sequence. Thus, the limiting probability generating function is unique. Probability generating function uniquely determines random variable. Therefore, the limit $\lambda_0$ is unique, and $\sum_{i=1}^{N} p^{\ast}_{i,t}\xrightarrow[N\to\infty]{d}\lambda_0$. Combined with uniform integrability \eqref{eq_p_i_tight}, the distributional convergence also implies $\mathbb{E}\sum_{i=1}^{N} p^{\ast}_{i,t}\xrightarrow[N\to\infty]{}\mathbb{E}\lambda_0$.
\end{proof}

\begin{proof}[Proof of Theorem \ref{th_aggregation_nonlin_big}]
We show by induction that, for all $T\geq-s+1$,
$\bigl[\sum_{i=1}^{N}p_{i,t},\,\sum_{i=1}^{N}y_{i,t}\bigr]_{t=-s+1,\ldots,T}$ converges both in distribution and in expectation to a nonlinear Poisson autoregression $\bigl[\lambda_t,X_t\bigr]_{t=-s+1,\ldots,T}$ and, additionally, $\max_i p_{i,T}\xrightarrow[N\to\infty]{p}0$. The proof follows the same steps as the proof of Theorem \ref{th_aggregation}, with the only difference being that we now have multiple nonlinear terms. At $t=-s+1,\ldots,0$, we are exactly in the setting of Lemma \ref{th_Poisson_randon}, so the desired (joint) convergence is satisfied. Moreover, applying conditional expectation, for $\tau=-s+1,\ldots,0$,
\begin{equation*}
\begin{split}
\mathbb{E}\sum\limits_{i=1}^N y_{i,\tau}&=\mathbb{E}\sum\limits_{i=1}^N p_{i,\tau}\xrightarrow[N\to\infty]{}\mathbb{E}\lambda_{\tau}=\mathbb{E}\bigl[\mathbb{E}(X_{\tau}\mid\lambda_{-s+1},X_{-s+1}\ldots,X_{\tau-1},\lambda_{\tau})\bigr]=\mathbb{E}X_{\tau}.
\end{split}
\end{equation*}
Let us show that $T=1$ follows from $T=0$. The argument for the general induction step is the same.

Let us show that $\max_i p_{i,1}\xrightarrow[N\to\infty]{p}0$. By Markov inequality,
$$\frac1{N^{m}}\sum\limits_{j=1}^{N} y_{j,\tau}\xrightarrow[N\to\infty]{p}0,\qquad\frac1{N^{m}}\sum\limits_{j=1}^{N} p_{j,\tau}\xrightarrow[N\to\infty]{p}0\qquad \tau=0,\ldots,-s+1,\,m>0.$$ Because
$$|f_{\alpha,i}(y_{i,t-1},\ldots,y_{i,t-s})|\leq A\sum_{\tau=1}^{s}y_{i,t-\tau}\leq sA,\qquad
{\sum_{i=1}^{N}|f_{\alpha,i}(y_{i,t-1},\ldots,y_{i,t-s})|\leq A\sum_{\tau=1}^{s}\sum_{i=1}^N y_{i,t-\tau}},$$
we obtain, uniformly in $i$,
\begin{equation}\label{eq_falpha_bound}
\frac1{N^{\kappa}}f_{\alpha,i}(y_{i,0},\ldots,y_{i,-s+1})\xrightarrow[N\to\infty]{p}0,\qquad
\frac1{N^{\kappa}}\sum\limits_{i=1}^{N}f_{\alpha,i}(y_{i,0},\ldots,y_{i,-s+1})\xrightarrow[N\to\infty]{p}0.
\end{equation}
Since $f_{\gamma,i}(0,\ldots,0)=0$ and $\nabla f_{\gamma,i}$ is uniformly bounded,
\begin{equation}\label{eq_fi_bound}
  \begin{split}
&\left|f_{\gamma,i}\left(\frac1{N}\sum\limits_{j=1}^{N}y_{j,t-1},\ldots,\frac1{N}\sum\limits_{j=1}^{N}y_{j,t-s},
\frac1{N}\sum\limits_{j=1}^{N}p_{j,t-1},\ldots,\frac1{N}\sum\limits_{j=1}^{N}p_{j,t-s}\right)\right|\\
&\leq M_1
\left(\frac1{N}\sum\limits_{j=1}^{N}y_{j,t-1}+\ldots+\frac1{N}\sum\limits_{j=1}^{N}y_{j,t-s}
+\frac1{N}\sum\limits_{j=1}^{N}p_{j,t-1}+\ldots+\frac1{N}\sum\limits_{j=1}^{N}p_{j,t-s}\right)\xrightarrow[N\to\infty]{p}0.
\end{split}\end{equation}
Finally, because $f_{\gamma}$ is continuous, $$f_{\gamma}\left(\sum\limits_{j=1}^{N}y_{j,0},\ldots,\sum\limits_{j=1}^{N}y_{j,-s+1},\sum\limits_{j=1}^{N}p_{j,0},\ldots,\sum\limits_{j=1}^{N}p_{j,-s+1}\right)
\xrightarrow[N\to\infty]{d} f_{\gamma}(X_0,\ldots,X_{-s+1},\lambda_0,\ldots,\lambda_{-s+1}),$$
so that combining with $\|\boldsymbol\gamma_i\|_1<C$, we get, uniformly in $i$,
\begin{equation}\label{eq_fgamma_bound}
  \begin{split}
&\left|\frac1{N}\boldsymbol\gamma_i^{\T}
f_{\gamma}\left(\sum\limits_{j=1}^{N}y_{j,0},\ldots,\sum\limits_{j=1}^{N}y_{j,-s+1},\sum\limits_{j=1}^{N}p_{j,0},\ldots,\sum\limits_{j=1}^{N}p_{j,-s+1}\right)\right|\\
&\leq \frac{C}{N}\left\| f_{\gamma}\left(\sum\limits_{j=1}^{N}y_{j,0},\ldots,\sum\limits_{j=1}^{N}y_{j,-s+1},\sum\limits_{j=1}^{N}p_{j,0},\ldots,\sum\limits_{j=1}^{N}p_{j,-s+1}\right) \right\|_1
\xrightarrow[N\to\infty]{p}0.
\end{split}\end{equation}
Combining Eq.~\eqref{eq_falpha_bound}, \eqref{eq_fi_bound}, \eqref{eq_fgamma_bound}, and the assumption on the initial conditions ${\max_i p_{i,\tau}\to0},\,\tau=0,\ldots,-s+1$, we deduce
$$0\leq\max_i p_{i,1}\xrightarrow[N\to\infty]{p}0.$$
Next, let us add $\sum\limits_{i=1}^N p_{i,1}$ to the joint convergence. Summing Eq.~\eqref{eq_multiple_k_diff_nonlin_big} with respect to $i$ for $t=1$, and using Eq.~\ref{eq_falpha_bound} and the fact that
\begin{equation*}
\begin{split}
&\sum\limits_{i=1}^N \|\nabla^2 f_{\gamma,i}\|\left(\frac1{N}\sum\limits_{j=1}^{N}y_{j,0}+\ldots+\frac1{N}\sum\limits_{j=1}^{N}y_{j,-s+1}
+\frac1{N}\sum\limits_{j=1}^{N}p_{j,0}+\ldots+\frac1{N}\sum\limits_{j=1}^{N}p_{j,-s+1}\right)^2\\
&\leq M_2\left(\sum\limits_{j=1}^{N}y_{j,0}+\ldots+\sum\limits_{j=1}^{N}y_{j,-s+1}
+\sum\limits_{j=1}^{N}p_{j,0}+\ldots+\sum\limits_{j=1}^{N}p_{j,-s+1}\right)\\
&\cdot\left(\frac1{N}\sum\limits_{j=1}^{N}y_{j,0}+\ldots+\frac1{N}\sum\limits_{j=1}^{N}y_{j,-s+1}
+\frac1{N}\sum\limits_{j=1}^{N}p_{j,0}+\ldots+\frac1{N}\sum\limits_{j=1}^{N}p_{j,-s+1}\right)\xrightarrow[N\to\infty]{p}0,
\end{split}
\end{equation*}
we obtain
\begin{equation*}
\begin{split}
\sum\limits_{i=1}^N p_{i,1}&=\frac1{N}\sum\limits_{i=1}^N c_{i}+\frac1{N^{\kappa}}\sum\limits_{i=1}^Nf_{\alpha,i}(y_{i,0},\ldots,y_{i,-s+1})
+\sum\limits_{\tau=1}^{s}\beta_{\tau} \sum\limits_{i=1}^Np_{i,1-\tau}\\
&+\sum\limits_{i=1}^Nf_{\gamma,i}\left(\frac1{N}\sum\limits_{j=1}^{N}y_{j,0},\ldots,\frac1{N}\sum\limits_{j=1}^{N}y_{j,-s+1}, \frac1{N}\sum\limits_{j=1}^{N}p_{j,0},\ldots,\frac1{N}\sum\limits_{j=1}^{N}p_{j,-s+1}\right)\\
&+\left(\frac1{N}\sum\limits_{i=1}^N\boldsymbol\gamma_i^{\T}\right)
f_{\gamma}\left(\sum\limits_{j=1}^{N}y_{j,0},\ldots,\sum\limits_{j=1}^{N}y_{j,-s+1}, \sum\limits_{j=1}^{N}p_{j,0},\ldots,\sum\limits_{j=1}^{N}p_{j,-s+1}\right)\\
&\xrightarrow[N\to\infty]{d} \bar{c}+\sum\limits_{\tau=1}^{s}\beta_{\tau}\lambda_{1-\tau}+(X_{0},\ldots,X_{-s+1},\lambda_{0},\ldots,\lambda_{-s+1})\bar{\mathbf{f}}\\
  &+\bar{\boldsymbol\gamma}^{\T}f_{\gamma}(X_{0},\ldots,X_{-s+1},\lambda_{0},\ldots\lambda_{-s+1})=\lambda_1
\end{split}
\end{equation*}
and
$$\left[\sum\limits_{i=1}^N p_{i,-s+1},\,\sum\limits_{i=1}^N y_{i,-s+1},\ldots,\sum\limits_{i=1}^N p_{i,0},\,\sum\limits_{i=1}^N y_{i,0},\,\sum\limits_{i=1}^N p_{i,1},\right]\xrightarrow[N\to\infty]{d} [\lambda_{-s+1},X_{-s+1},\ldots,\lambda_0,X_0,\lambda_1].$$
Direct application of Lemma \ref{th_Poisson_randon} with $q_i=p_{i,1}$ gives
\begin{equation}\label{eq_agg_t1_nonlin}
\left[\sum\limits_{i=1}^N p_{i,1},\,\sum\limits_{i=1}^N y_{i,1}\right]\xrightarrow[N\to\infty]{d} [\lambda_1,X_1],~{\rm with}~X_1\sim
Pois(\lambda_1).
\end{equation}
By adding $\{p_{i,t},y_{i,t}\}_{i=1,\ldots,N,\,t=-s+1,\ldots,0}$ to the conditioning in \eqref{eq_G_N(z)}, and repeating the same proof, the convergence in \eqref{eq_agg_t1_nonlin} extends to conditional convergence. Hence, we arrive at
$$
\left[\sum\limits_{i=1}^N p_{i,t},\,\sum\limits_{i=1}^N y_{i,t}\right]_{t=-s+1,\ldots,1}\xrightarrow[N\to\infty]{d} [\lambda_t,X_t]_{t=-s+1,\ldots,1}.
$$
It remains to prove convergence in expectations. For $t\leq0$ we have already shown it. Let us separately analyze all nonlinear terms in $t=1$ dynamics. First,
\begin{equation}\label{eq_falpha_bound_E}
0\leq\frac1{N^{\kappa}}\sum\limits_{i=1}^N\mathbb{E}f_{\alpha,i}(y_{i,0},\ldots,y_{i,-s+1})\leq\frac{A}{N^{\kappa}}\sum\limits_{\tau=1}^{s}\sum_{i=1}^N \mathbb{E}y_{i,1-\tau}
\xrightarrow[N\to\infty]{}0.
\end{equation}
Next, let $f_{\gamma}=\left(f_{\gamma}^{[1]},\ldots,f_{\gamma}^{[k]}\right)^{\T}$ and consider $r=1,\ldots,k$,
\begin{equation*}
0\leq f_{\gamma}^{[r]}\left(\sum\limits_{j=1}^{N}y_{j,0},\ldots,\sum\limits_{j=1}^{N}y_{j,-s+1}, \sum\limits_{j=1}^{N}p_{j,0},\ldots,\sum\limits_{j=1}^{N}p_{j,-s+1}\right)
\leq L_{\gamma}\sum\limits_{\tau=1}^{s}\left(\sum\limits_{j=1}^{N}y_{j,1-\tau}+\sum\limits_{j=1}^{N}p_{j,1-\tau}\right)
\end{equation*}
Each term on the right converges in distribution and in expectation to its nonnegative limit and is therefore uniformly integrable (for nonnegative random variables, convergence in distribution together with convergence of first moments implies uniform integrability). Therefore, $f_{\gamma}^{[r]}\left(\sum\limits_{j=1}^{N}y_{j,0},\ldots,\sum\limits_{j=1}^{N}y_{j,-s+1}, \sum\limits_{j=1}^{N}p_{j,0},\ldots,\sum\limits_{j=1}^{N}p_{j,-s+1}\right)$ is also uniformly integrable. Combined with convergence in distribution, this gives convergence of expectations:
\begin{equation}\label{eq_fgamma_bound_E}
\begin{split}
&\mathbb{E}f_{\gamma}^{[r]}\left(\sum\limits_{j=1}^{N}y_{j,0},\ldots,\sum\limits_{j=1}^{N}y_{j,-s+1}, \sum\limits_{j=1}^{N}p_{j,0},\ldots,\sum\limits_{j=1}^{N}p_{j,-s+1}\right)\\
&\xrightarrow[N\to\infty]{}\mathbb{E}f_{\gamma}^{[r]}(X_{0},\ldots,X_{-s+1},\lambda_{0},\ldots\lambda_{-s+1}).
\end{split}\end{equation}
It remains to consider $f_{\gamma,i}$. A second-order Taylor expansion around zero gives
\begin{equation*}
\begin{split}
&\left| \sum_{i=1}^N \mathbb{E}f_{\gamma,i}\left(\frac1{N}\sum\limits_{j=1}^{N}y_{j,0},\ldots,\frac1{N}\sum\limits_{j=1}^{N}y_{j,-s+1}, \frac1{N}\sum\limits_{j=1}^{N}p_{j,0},\ldots,\frac1{N}\sum\limits_{j=1}^{N}p_{j,-s+1}\right) \right.\\
-&\left.
\mathbb{E}\left[\left(\sum_{j=1}^{N}y_{j,0},\ldots,\sum_{j=1}^{N}y_{j,-s+1},\sum_{j=1}^{N}p_{j,0},\ldots,\sum_{j=1}^{N}p_{j,-s+1}\right)
\frac1N\sum_{i=1}^N\nabla f_{\gamma,i}(0,\ldots,0)\right]\right|\\
\leq&\frac{M_2}{2N}\mathbb{E}\left[\left(\sum_{j=1}^Ny_{j,0}\right)^2+\ldots+\left(\sum_{j=1}^Ny_{j,-s+1}\right)^2
+\left(\sum_{j=1}^Np_{j,0}\right)^2+\ldots+\left(\sum_{j=1}^Np_{j,-s+1}\right)^2\right].
\end{split}
\end{equation*}
Since each sum over $j$ is bounded by $N$, for every $\eps>0$,
\begin{equation*}
\begin{split}
\frac1{N}\mathbb{E}\left(\sum\limits_{i=1}^N y_{i,\tau}\right)^2
&\leq\eps\mathbb{E}\sum_{i=1}^N y_{i,\tau}+\mathbb{E}\left[\sum_{i=1}^N y_{i,\tau}\1\left(\sum_{i=1}^N y_{i,\tau}>\eps N\right)\right],\\
\frac1{N}\mathbb{E}\left(\sum\limits_{i=1}^N p_{i,\tau}\right)^2
&\leq\eps\mathbb{E}\sum_{i=1}^N p_{i,\tau}+\mathbb{E}\left[\sum_{i=1}^N p_{i,\tau}\1\left(\sum_{i=1}^N p_{i,\tau}>\eps N\right)\right].
\end{split}
\end{equation*}
The second terms converge to zero as $N\to\infty$ by uniform integrability, while the first terms are bounded by a constant times $\eps$. Thus, as $N\to\infty$ followed by $\eps\to0$, the second-order Taylor remainder goes to $0$ and we obtain
\begin{equation}\label{eq_fgammai_bound_E}
\begin{split}
&\sum_{i=1}^N \mathbb{E}f_{\gamma,i}\left(\frac1{N}\sum\limits_{j=1}^{N}y_{j,0},\ldots,\frac1{N}\sum\limits_{j=1}^{N}y_{j,-s+1}, \frac1{N}\sum\limits_{j=1}^{N}p_{j,0},\ldots,\frac1{N}\sum\limits_{j=1}^{N}p_{j,-s+1}\right)\\
&\xrightarrow[N\to\infty]{}\mathbb{E}(X_{0},\ldots,X_{-s+1},\lambda_{0},\ldots,\lambda_{-s+1})\bar{\mathbf{f}}.
\end{split}\end{equation}
Summing up, Eq.~\eqref{eq_falpha_bound_E}, \eqref{eq_fgamma_bound_E}, \eqref{eq_fgammai_bound_E} imply
\begin{equation*}\begin{split}
\sum\limits_{i=1}^N \mathbb{E}p_{i,1}&
=\frac{1}{N}\sum\limits_{i=1}^N c_{i}+\frac1{N^{\kappa}}\sum\limits_{i=1}^N\mathbb{E}f_{\alpha,i}(y_{i,0},\ldots,y_{i,-s+1})
+\sum\limits_{\tau=1}^{s}\beta_{\tau} \sum\limits_{i=1}^N\mathbb{E}p_{i,1-\tau}\\
&+\sum\limits_{i=1}^N\mathbb{E}f_{\gamma,i}\left(\frac1{N}\sum\limits_{j=1}^{N}y_{j,0},\ldots,\frac1{N}\sum\limits_{j=1}^{N}y_{j,-s+1}, \frac1{N}\sum\limits_{j=1}^{N}p_{j,0},\ldots,\frac1{N}\sum\limits_{j=1}^{N}p_{j,-s+1}\right)\\
&+\left(\frac1{N}\sum\limits_{i=1}^N\boldsymbol\gamma_i^{\T}\right)
\mathbb{E}f_{\gamma}\left(\sum\limits_{j=1}^{N}y_{j,0},\ldots,\sum\limits_{j=1}^{N}y_{j,-s+1}, \sum\limits_{j=1}^{N}p_{j,0},\ldots,\sum\limits_{j=1}^{N}p_{j,-s+1}\right)\\
&\xrightarrow[N\to\infty]{} \bar{c}+\sum\limits_{\tau=1}^{s}\beta_{\tau}\mathbb{E}\lambda_{1-\tau}+\mathbb{E}(X_{0},\ldots,X_{-s+1},\lambda_{0},\ldots,\lambda_{-s+1})\bar{\mathbf{f}}\\
  &+\bar{\boldsymbol\gamma}^{\T}\mathbb{E}f_{\gamma}(X_{0},\ldots,X_{-s+1},\lambda_{0},\ldots\lambda_{-s+1})=\mathbb{E}\lambda_1.
\end{split}\end{equation*}
Finally,
$$\sum\limits_{i=1}^N \mathbb{E}y_{i,1}=\sum\limits_{i=1}^N \mathbb{E}p_{i,1}\xrightarrow[N\to\infty]{}\mathbb{E}\lambda_1
=\mathbb{E}\left[\mathbb{E}\left(X_1\mid \lambda_{-s+1},X_{-s+1},\ldots,\lambda_0,X_0,\lambda_1\right)\right]=\mathbb{E}X_1.$$
Repeating the same steps for general $T$, we arrive at the desired joint convergence.
\end{proof}

\begin{proof}[Proof of Theorem \ref{th_aggregation_network}]
The proof follows the lines of the previous proofs and proceeds by induction. The difference is that $\max_i p_{i,t}$ involves the term $\sum\limits_{j=1}^N W_{ij}y_{j,t-1}$. We are going to show that this term still goes to $0$ in probability.
$$\sum\limits_{j=1}^N W_{ij}y_{j,t-1}\leq\frac1{d_i^{out}}\sum\limits_{j=1}^N y_{j,t-1}\leq\frac1{d\log N}\sum\limits_{j=1}^N y_{j,t-1}\xrightarrow[N\to\infty]{p}0,$$
since $\sum\limits_{j=1}^N y_{j,t-1}=O_p(1)$ by induction arguments.

Second, note that because $W_{ij}$ is doubly stochastic, $\sum\limits_{i=1}^{N}W_{ij}=1$ and
$$\sum\limits_{i=1}^{N}\sum\limits_{j=1}^N W_{ij}y_{j,t-1}=\sum\limits_{j=1}^N y_{j,t-1}\sum\limits_{i=1}^{N}W_{ij}=\sum\limits_{j=1}^N y_{j,t-1}.$$
Thus,
$$\sum_{i=1}^N p_{i,t}=\frac1{N}\sum_{i=1}^N c_i+\frac1{N^{\kappa}}\sum_{i=1}^N a_i y_{i,t-1}+\beta\sum_{i=1}^N p_{i,t-1}+\gamma\sum_{i=1}^N y_{i,t-1}.$$
Therefore, repeating the steps of the proof of Theorem \ref{th_aggregation}, we obtain the desired result.
\end{proof}

\section{Proofs: Estimation}\label{append_estimation}

\begin{proof}[Proof of Theorem \ref{th_consistency}] The proofs first defines the population analogue of the criterion function, then shows that the population criterion is maximized by $\theta_0$, and establishes the uniform convergence of the sample criterion to the population one. Finally, these three steps are combined to imply consistency.

Start by noting that both Assumptions \ref{ass_contraction1} and \ref{ass_contraction2} imply that if we choose different starting values at time $r$ to initialize probabilities, but then use exactly the same realizations of $\y_t$ for two dynamics with the same parameter $\theta$, then there exists $\varrho\in[0,1)$ such that
\begin{equation}\label{eq_forget_past}
 \sup_{\theta\in\Theta} \left\| \p_t(\theta)-\p'_t(\theta)\right\|_1 \leq C\varrho^{t-r}, \qquad t\geq r,
\end{equation}
where the constant $C$ can be chosen uniformly over the initial values since all probability coordinates belong to $[0,1]$.

To define a proper population object, consider a unique stationary solution $\{\y_t,\p_t\}_{t\in\mathbb{Z}}$ of Theorem \ref{th_stationary_unique} with corresponding to the true parameter $\theta_0$. for every $\theta\in\Theta$, consider the recursion
\begin{equation}\label{eq_p_recursion}
 \p_t^*(\theta) = \g\left( \{\p_{\tau}^*(\theta)\}_{\tau=t-1,\ldots,t-s}, \{\y_{\tau}\}_{\tau=t-1,\ldots,t-q}; \theta \right),
\end{equation}
where the recursion is driven by the observations generated under the true parameter $\theta_0$. Let us show that there is a unique stationary solution to the recursion \eqref{eq_p_recursion}. Fix arbitrary initial values $\p^0\in[0,1]^{Ns}$ and consider the process $\p_t^{[r]}(\theta)$, which at time $t=r$ starts the recursion \eqref{eq_p_recursion} from $\p^0$, i.e.,
$$\p_{r}^{[r]}(\theta)=\g\left( \p^0, \{\y_{\tau}\}_{\tau=t-1,\ldots,t-q}; \theta \right),$$
and for $t>r$ continues as in the recursion by plugging $\p_t^{[r]}(\theta)$ in the right-hand side of the recursion (for $r<t<r+s$ use chosen initial values instead of the unobserved lags $\p_{t-t'}^{[r]}(\theta),\,t-r<t'\leq s$). Then for $r'<r$, starting from $t=r$, $\p_t^{[r']}(\theta)$ and $\p_t^{[r]}(\theta)$ are driven by exactly the same observations $\y_t$. Thus, \eqref{eq_forget_past} implies
$$\sup_{\theta\in\Theta} \left\| \p_t^{[r]}(\theta)-\p_t^{[r']}(\theta) \right\|_1 \leq C\varrho^{t-r}.$$
Therefore, $\p_t^{[r]}(\theta)$ is uniformly Cauchy as $r\to-\infty$ and we can define, for every realization of $\{\y_{\tau}\}$
\begin{equation}\label{eq_stationary_p_theta}
 \p_t^*(\theta)= \lim_{r\to-\infty}\p_t^{[r]}(\theta).
\end{equation}
The convergence is uniform over $\theta\in\Theta$, and the limit is independent of $\p^0$ (again, due to \eqref{eq_forget_past}). Since $\g$ is continuous and $\p_t^{[r]}(\theta)$ satisfies the recursion \eqref{eq_p_recursion} for $t>r+s$, taking the limits, we obtain that $\p_t^*(\theta)$ satisfies the desired recursion for all $t$. Moreover, $\p_t^*(\theta)$ is measurable with respect to $\sigma(\y_{t-1},\y_{t-2},\ldots)$.

At the true parameter, $\p_t^*(\theta_0)=\p_t$, since both $\p_t$ and $\p_t^{[r]}(\theta_0)$ are driven by the same realizations $\y_{\tau}$ for $t>r$ and \eqref{eq_forget_past} implies
$$\left\| \p_t^{[r]}(\theta_0)-\p_t \right\|_1 \leq C\varrho^{t-r}.$$
Taking the limit $r\to-\infty$, we obtain $\p_t^*(\theta_0)=\p_t$.

Finally, because $\g$ satisfies either Assumptions \ref{ass_contraction1} or \ref{ass_contraction2}, it is continuous function of $\p_t$ and $\y_t$ (uniformly over $\theta$). Thus, $\p_{t}^{[r]}(\theta)$ is a continuous function of $\theta$ for fixed $t>r$. Since $\p_{t}^{[r]}(\theta)\to \p_{t}^{*}(\theta)$ uniformly over $\theta$, $\p_{t}^{*}(\theta)$ is also continuous function of $\theta$.

\smallskip

To show population identification of the parameter, define
\begin{equation}\label{eq_mle_pstar}
 \ell_t^*(\theta)=\sum_{i=1}^N \left[ y_{i,t}\log p_{i,t}^*(\theta) +(1-y_{i,t})\log\bigl(1-p_{i,t}^*(\theta)\bigr) \right]
\end{equation}
and the population objective
\begin{equation}\label{eq_population_mle}
 Q(\theta):=\mathbb{E}\ell_t^*(\theta).
\end{equation}
By Assumption \ref{ass_g_theta}\ref{enu:g_theta:bounded},
\begin{equation}\label{eq_loglik_bound_mle}
 \sup_{\theta\in\Theta}|\ell_t^*(\theta)| \leq N|\log\eps|.
\end{equation}

Let us show that $Q(\theta)$ has a unique maximum at $\theta_0$. Consider
\begin{equation}\label{eq_maxQ}
\begin{split}
Q(\theta)-Q(\theta_0)
&=\sum\limits_{i=1}^{N}\mathbb{E}\left[ y_{i,t}\log \frac{p_{i,t}^*(\theta)}{p_{i,t}} +(1-y_{i,t}) \log \frac{1-p_{i,t}^*(\theta)}{1-p_{i,t}} \right]\\
&=\sum\limits_{i=1}^{N}\mathbb{E}\log\frac{(p_{i,t}^*(\theta))^{y_{i,t}}(1-p_{i,t}^*(\theta))^{1-y_{i,t}}}{p_{i,t}^{y_{i,t}}(1-p_{i,t})^{1-y_{i,t}}}
\leq\sum\limits_{i=1}^{N}\mathbb{E}\frac{(p_{i,t}^*(\theta))^{y_{i,t}}(1-p_{i,t}^*(\theta))^{1-y_{i,t}}}{p_{i,t}^{y_{i,t}}(1-p_{i,t})^{1-y_{i,t}}}-N\\
&=\sum\limits_{i=1}^{N}\mathbb{E}\mathbb{E}_{t-1}\frac{(p_{i,t}^*(\theta))^{y_{i,t}}(1-p_{i,t}^*(\theta))^{1-y_{i,t}}}{p_{i,t}^{y_{i,t}}(1-p_{i,t})^{1-y_{i,t}}}-N\\
&=\sum\limits_{i=1}^{N}\mathbb{E}\left(\frac{p_{i,t}^*(\theta)}{p_{i,t}}p_{i,t}+\frac{1-p_{i,t}^*(\theta)}{1-p_{i,t}}(1-p_{i,t})\right)-N\equiv N-N=0,
\end{split}\end{equation}
where we used the inequality $\log(a)\leq a-1$ for any $a\geq0$. Thus, $\theta_0$ is a maximum of $Q(\theta)$. Because of \eqref{eq_loglik_bound_mle}, the value of $Q(\theta_0)$ is finite. Moreover, $\theta_0$ is the unique maximum of $Q(\theta)$, because equality in \eqref{eq_maxQ} is only possible if $\frac{p_{i,t}^*(\theta)^{y_{i,t}}(1-p_{i,t}^*(\theta))^{1-y_{i,t}}}{p_{i,t}^{y_{i,t}}(1-p_{i,t})^{1-y_{i,t}}}=1$ with probability $1$. So equality requires $p_{i,t}^*(\theta)=p_{i,t}$ for all $i$ (because, conditional on $\F_{t-1}$, $y_{i,t}$ takes both values $0$ and $1$ with positive probability under Assumption \ref{ass_g_theta}\ref{enu:g_theta:bounded}). This cannot hold by Assumption \ref{ass_g_theta}\ref{enu:g_theta:unique} if $\theta \ne \theta_0$.

Finally, because $\p_{t}^*(\theta)$ is a continuous function of $\theta$ and $g_{i}(\cdot)\in[\eps,1-\eps]$, for any $\delta>0$,
\begin{equation}\label{eq_ID_MLE}
\sup_{\theta\in\Theta\setminus B(\theta_0,\delta)}Q(\theta)<Q(\theta_0),
\end{equation}
where $ B(\theta_0,\delta)$ denotes an open ball of radius $\delta$ around $\theta_0$.

\smallskip

Next, we need to verify the uniform convergence of the sample criterion to the population one:
\begin{equation}\label{eq_UC_mle}
\sup_{\theta\in\Theta} \left| \frac1T\sum_{t=1}^T\ell_t^*(\theta)-Q(\theta) \right|\xrightarrow[T\to\infty]{p}0.
\end{equation}
To establish \eqref{eq_UC_mle}, fix an integer $M\geq1$ and consider $\p_t^{[t-M]}(\theta),$ that is, the probability filter initialized $M$ periods before $t$ at the
fixed values $\p^0$. Define correspondingly
$$\ell_{t,M}(\theta)=\sum_{i=1}^N  \left[y_{i,t}\log p_{i,t}^{[t-M]}(\theta) +(1-y_{i,t})\log\left(1-p_{i,t}^{[t-M]}(\theta)\right)\right],
    \qquad    Q_M(\theta):=\mathbb{E}\ell_{t,M}(\theta).$$
By \eqref{eq_forget_past},
$$ \sup_{\theta\in\Theta}  \|\p_t^{[t-M]}(\theta)-\p_t^*(\theta)\|_1  \leq C\varrho^M.$$
Combining the above bound with Assumption \ref{ass_g_theta}\ref{enu:g_theta:bounded},
\begin{equation}\label{eq_truncation_likelihood}
    \sup_{\theta\in\Theta}|\ell_{t,M}(\theta)-\ell_t^*(\theta)| \leq \frac{C}{\eps}\varrho^M,
\end{equation}
and hence
\begin{equation}\label{eq_truncation_population}
    \sup_{\theta\in\Theta} |Q_M(\theta)-Q(\theta)|\leq \frac{C}{\eps}\varrho^M.
\end{equation}
For fixed $M$, $\ell_{t,M}(\theta)$ is a continuous function of $\theta$ and of a fixed finite block of the stationary Markov process $\{Z_t\}$. Therefore, by Corollary \ref{cor_lln} applied to the corresponding finite-lag augmentation of $Z_t$, for every fixed $\theta\in\Theta$,
\begin{equation}\label{eq_pointwise_truncated_lln}
    \frac1T\sum_{t=1}^T\ell_{t,M}(\theta) \xrightarrow[T\to\infty]{p} Q_M(\theta).
\end{equation}
The convergence is uniform in $\theta$, since $\ell_{t,M}(\theta)$ is continuous and the state-space of $Z_t$ and $\Theta$ are compact, implying uniform continuity:
\begin{equation}\label{eq_truncated_ulln}
    \sup_{\theta\in\Theta} \left|\frac1T\sum_{t=1}^T\ell_{t,M}(\theta)-Q_M(\theta)\right|
    \xrightarrow[T\to\infty]{p}0
\end{equation}
for every fixed $M$.

Finally, \eqref{eq_truncation_likelihood}--\eqref{eq_truncated_ulln} give \eqref{eq_UC_mle},
$$  \sup_{\theta\in\Theta} \left|\frac1T\sum_{t=1}^T\ell_t^*(\theta)-Q(\theta)\right|
    \leq \sup_{\theta\in\Theta}  \left|\frac1T\sum_{t=1}^T\ell_{t,M}(\theta)-Q_M(\theta)\right|+\frac{2C}{\eps}\varrho^M\to0,
$$
as first $T\to\infty$, and then $M\to\infty$.

To finish the proof, note that $\tilde{\p}_t(\theta)$ and $\p_t^*(\theta)$ are driven by the same observed sequence $\{\y_t\}$ and differ only in their probability initialization, so bound \eqref{eq_forget_past} gives
\begin{equation}\label{eq_initialized_filter_bound}
 \sup_{\theta\in\Theta} \|\tilde{\p}_t(\theta)-\p_t^*(\theta)\|_1 \leq C\varrho^t.
\end{equation}
Assumption \ref{ass_g_theta}\ref{enu:g_theta:bounded} implies that for $\ell_{i,t}(p)= y_{i,t}\log p+(1-y_{i,t})\log(1-p)$,
$$\sup_{p\in[\eps,1-\eps]} \left|\frac{\partial\ell_{i,t}(p)}{\partial p}\right| \leq\frac1{\eps}.$$
Combining the above bound with \eqref{eq_initialized_filter_bound}, we get
$$\sup_{\theta\in\Theta} \left| \frac1T\mathcal Q_T(\theta) - \frac1T\sum_{t=1}^T\ell_t^*(\theta) \right|\leq \frac{C}{T\eps}\sum_{t=1}^T\varrho^t
 \leq \frac{C}{T\eps(1-\varrho)}\xrightarrow[T\to\infty]{}0,$$
which together with \eqref{eq_UC_mle} implies
\begin{equation}\label{eq_UC_mle_tilde}
\sup_{\theta\in\Theta} \left| \frac1T\mathcal{Q}_T(\theta)-Q(\theta) \right|\xrightarrow[T\to\infty]{p}0.
\end{equation}

Identification \eqref{eq_ID_MLE} and uniform convergence \eqref{eq_UC_mle_tilde} imply that for any $\delta>0$, ${Prob(\|\hat{\theta}^{MLE}-\theta_0\|>\delta)\xrightarrow[T\to\infty]{}0}$. That is, $\hat{\theta}^{MLE}\xrightarrow[T\to\infty]{p}\theta_0$.
\end{proof}

\begin{proof}[Proof of Theorem \ref{th_asy_normality}] The proof starts with defining first and second order derivatives, which appear in the score and the Hessian. Then the CLT for the score and LLN for the Hessian are derived and combined to obtain asymptotic normality of $\hat{\theta}^{MLE}$.

Let $\p_t^*(\theta)$ denote the stationary probability, which satisfies recursion \eqref{eq_p_recursion} constructed with the true realizations $\y_t$, but different $\theta$, as in the proof of Theorem \ref{th_consistency}. We first need to show that $\p_t^*(\theta)$ is twice continuously differentiable. For this purpose, consider its finite-past approximations $\p_t^{[r]}(\theta)$ from fixed and independent of $\theta$ initial values $\p^0\in[0,1]^{Ns}$ and define
$$G_t(P;\theta):=\g\left(P, \{\y_{\tau}\}_{\tau=t-1,\ldots,t-q};\theta\right),$$
where $P$ stacks the $s$ probability lags. Since $G_t$ is twice continuously differentiable, $\p_t^{[r]}(\theta)$ is twice continuously differentiable in $\theta$. Writing its first derivative as
an $N\times m$ Jacobian $D_t^{[r]}(\theta)$, differentiation of the filter recursion gives
\begin{equation}\label{eq_derivative_filter}
\begin{split}
    D_t^{[r]}(\theta)    &=
    G_{\theta,t}^{[r]}(\theta)  +\sum_{j=1}^s G_{p_j,t}^{[r]}(\theta)D_{t-j}^{[r]}(\theta),\quad t\geq r,\\
    D_{r-j}^{[r]}(\theta)&=0,\quad j=1,\ldots,s,
 \end{split}
\end{equation}
where the derivatives of $G_t$ are evaluated along the finite-past filter $\p_t^{[r]}(\theta)$. The derivative recursion inherits the contraction property of Assumptions \ref{ass_contraction1} and \ref{ass_contraction2}. Under \ref{ass_contraction1}, the absolute values of the entries of $G_{p_j,t}^{[r]}(\theta)$ are bounded by the corresponding probability Lipschitz coefficients. Hence, after stacking the $s$ derivative lags, the homogeneous part of \eqref{eq_derivative_filter} is dominated by the companion matrix in Assumption \ref{ass_contraction1}, whose spectral radius is strictly below one. Under Assumption \ref{ass_contraction2}, the contractive bound is applied to the $\ell_1$-norms of each of $m$ columns (and, thus, to their maximum). Thus, under either Assumption \ref{ass_contraction1} or \ref{ass_contraction2},
\begin{equation}\label{eq_MLE_deriv_bound}
  \sup_{r\in\mathbb{Z}}\sup_{t\geq r}\sup_{\theta\in\Theta_0}\|D_t^{[r]}(\theta)\|_1<\infty.
\end{equation}
If $r'<r$, subtract the derivative recursions corresponding to $\p_t^{[r]}(\theta)$ and $\p_t^{[r']}(\theta)$. Writing
$$ \Delta_t(\theta) :=  D_t^{[r]}(\theta)-D_t^{[r']}(\theta),$$
we obtain
$$   \Delta_t(\theta)= \sum_{j=1}^s G_{p_j,t}^{[r]}(\theta)\Delta_{t-j}(\theta)  +R_t(\theta),$$
where
$$ R_t(\theta)= G_{\theta,t}^{[r]}(\theta)-G_{\theta,t}^{[r']}(\theta)
  +\sum_{j=1}^s \left(G_{p_j,t}^{[r]}(\theta) - G_{p_j,t}^{[r']}(\theta)\right) D_{t-j}^{[r']}(\theta).$$
Since the second derivatives of $G_t$ are uniformly bounded on the compact set $\Theta_0$, its first derivatives are uniformly
Lipschitz in the probability arguments. Combining this with \eqref{eq_forget_past} and the uniform boundedness of
$D_t^{[r']}(\theta)$ gives
\[
    \sup_{\theta\in\Theta_0}\|R_t(\theta)\|
    \leq C\varrho^{\,t-r}.
\]
Since $\Delta_t(\theta)$ satisfies the same geometrically contracting recursion as the derivative filter itself (modulo different intercept), there exist $C<\infty$ and
$\rho_1\in[0,1)$, such that
\begin{equation*}
    \sup_{\theta\in\Theta_0} \|D_t^{[r]}(\theta)-D_t^{[r']}(\theta)\| \leq C\rho_1^{\,t-r}.
\end{equation*}
Thus the derivatives are uniformly Cauchy as $r\to-\infty$. Since
$\p_t^{[r]}(\theta)\to\p_t^*(\theta)$ uniformly in $\theta$, $\p_t^*(\theta)$ is continuously differentiable and
\[
    D_t^*(\theta):=\frac{\partial\p_t^*(\theta)}{\partial\theta^{\T}} =\lim_{r\to-\infty}D_t^{[r]}(\theta)
    =\lim_{r\to-\infty}\frac{\partial\p_t^{[r]}(\theta)}{\partial\theta^{\T}}
\]
uniformly over $\Theta_0$. Taking the limit in \eqref{eq_derivative_filter},
\begin{equation}\label{eq_derivative_filter}
    D_t^*(\theta)   =
    G_{\theta,t}^*(\theta)+\sum_{j=1}^s G_{p_j,t}^*(\theta)D_{t-j}^*(\theta).
\end{equation}
In particular, at $\theta=\theta_0$, because $\p_t^*(\theta_0)=\p_t$, \eqref{eq_derivative_filter} coincides with the recursion in Assumption \ref{ass_g_theta2}\ref{enu:g_theta2:H}.

Differentiating \eqref{eq_derivative_filter} once more and using the same arguments as before, where second derivatives of $\g$ are continuous and bounded by Assumption \ref{ass_g_theta2}\ref{enu:g_theta2:differ}), shows that the second derivative of $\p_t^*(\theta)$ exist and is uniformly bounded on $\Theta_0$:
$$\sup_t\sup_{\theta\in\Theta_0} \left\|\frac{\partial^2\p_t^*(\theta)}{\partial\theta\partial\theta^{\T}}\right\|<\infty.$$
The same recursions give forgetting of the derivatives for the initialized filter $\tilde{\p}_t(\theta)$ used in the likelihood. If
$\tilde D_t(\theta)=\partial\tilde{\p}_t(\theta)/\partial\theta^{\T}$, then, since $G_t$ is twice differentiable, so its first derivative is Lipschitz, there exist $C<\infty$ and $\rho_2\in[0,1)$,
\begin{equation*}
    \sup_{\theta\in\Theta_0}\|\tilde D_t(\theta)-D_t^*(\theta)\|\leq C\rho_2^t.
\end{equation*}
Similarly, a (weaker) bound on the second derivative holds:
\begin{equation*}
    \sup_{\theta\in\Theta_0}\left\|\frac{\partial^2\tilde{\p}_t(\theta)}{\partial\theta\partial\theta^{\T}}
    -\frac{\partial^2\p_t^*(\theta)}{\partial\theta\partial\theta^{\T}}\right\|\xrightarrow[t\to\infty]{a.s.}0.
\end{equation*}

\smallskip

Now we can analyze the score and the Hessian of the MLE problem. Recall the log-likelihood in \eqref{eq_mle_pstar},
$$\ell_t^*(\theta)=\sum_{i=1}^N \left[ y_{i,t}\log p_{i,t}^*(\theta) +(1-y_{i,t})\log\bigl(1-p_{i,t}^*(\theta)\bigr) \right]$$
and define the score function
$$s_t^*(\theta)=\frac{\partial\ell_t^*(\theta)}{\partial\theta}=\sum_{i=1}^N \left[\frac{y_{i,t}}{p_{i,t}^*(\theta)} - \frac{1-y_{i,t}}{1-p_{i,t}^*(\theta)} \right]\frac{\partial}{\partial\theta}p_{i,t}^*(\theta)
=\sum_{i=1}^N \frac{y_{i,t}-p_{i,t}^*(\theta)}{p_{i,t}^*(\theta)(1-p_{i,t}^*(\theta))}\frac{\partial}{\partial\theta}p_{i,t}^*(\theta).$$
Since both $p_{i,t}^*$ and $\frac{\partial}{\partial\theta}p_{i,t}^*(\theta)$ are $\F_{t-1}$ measurable, $p_{i,t}^*(\theta_0)=p_{i,t}$, and $\mathbb{E}_{t-1}y_{i,t}=p_{i,t}$, $\{s_t^*(\theta_0),\F_t\}$ is a martingale difference sequence. Its conditional variance is
\begin{equation*}\begin{split}
V_t:=\mathbb{E}_{t-1}s_t(\theta_0)s_{t}(\theta_0)^{\T}&=
\sum_{i=1}^{N}\frac{\partial}{\partial\theta}p_{i,t}^*(\theta_0)\frac{\partial}{\partial\theta^{\T}}p_{i,t}^*(\theta_0)
\mathbb{E}_{t-1}\left(\frac{y_{i,t}-p_{i,t}}{p_{i,t}(1-p_{i,t})}\right)^2\\
&+\sum_{i\neq j}\frac{\partial}{\partial\theta}p_{i,t}^*(\theta_0)\frac{\partial}{\partial\theta^{\T}}p_{j,t}^*(\theta_0)
\mathbb{E}_{t-1}\left(\frac{y_{i,t}-p_{i,t}}{p_{i,t}(1-p_{i,t})}\right)\left(\frac{y_{j,t}-p_{j,t}}{p_{j,t}(1-p_{j,t})}\right)
\\
&=\sum_{i=1}^{N}\frac{1}{p_{i,t}(1-p_{i,t})}\frac{\partial}{\partial\theta}p_{i,t}^*(\theta_0)\frac{\partial}{\partial\theta^{\T}}p_{i,t}^*(\theta_0).
\end{split}\end{equation*}
Let us show that
\begin{equation}\label{eq_MLE_var_LLN}
  \frac1T\sum_{t=1}^T V_t\xrightarrow[T\to\infty]{p}H_0.
\end{equation}
We proceed along the lines of the LLN part of the proof of Theorem \ref{th_consistency}. Fix $M$ and initialize the probability recursion and its derivative recursion at time $t-M$ from fixed values $\p^0$. Then define
\[
    V_{t,M}=\sum_{i=1}^N \frac{\left(D^{[t-M]}_{i,t}\right)^{\T}D^{[t-M]}_{i,t}}{p^{[t-M]}_{i,t}\left(1-p^{[t-M]}_{i,t}\right)}.
\]
Filter and its derivative past-forgetting imply
\[
    \sup_t\|V_{t,M}-V_t\|\xrightarrow[M\to\infty]{}0.
\]
For every fixed $M$, $V_{t,M}$ is a continuous function of a fixed finite block of the stationary Markov process $\{Z_t\}$. Therefore, Corollary \ref{cor_lln}, applied entrywise and to the corresponding finite-lag augmentation of $Z_t$, gives
\[
    \frac{1}{T}\sum_{t=1}^T V_{t,M}\xrightarrow[T\to\infty]{p}\mathbb{E} V_{t,M}.
\]
Letting first $T\to\infty$ for fixed $M$, and then $M\to\infty$, proves LLN \eqref{eq_MLE_var_LLN}, since $\mathbb{E}V_t=H_0$.

Assumption \ref{ass_g_theta}\ref{enu:g_theta:bounded} and the boundedness of the derivative \eqref{eq_MLE_deriv_bound} imply that $\| s_t^*(\theta_0)\|\leq C$. Hence the conditional Lindeberg condition is trivially satisfied and the martingale CLT gives
\begin{equation*}
    \frac1{\sqrt T}\sum_{t=1}^T s_t^*(\theta_0)\xrightarrow[T\to\infty]{d}\mathcal{N}(0,H_0).
\end{equation*}
For the initialized likelihood, let
\[
    \tilde{s}_t(\theta)=\frac{\partial}{\partial\theta}\sum_{i=1}^N \left[
       y_{i,t}\log\tilde p_{i,t}(\theta) +(1-y_{i,t})\log(1-\tilde p_{i,t}(\theta))\right].
\]
The probability and derivative forgetting bounds, together with the fact
that probabilities are bounded away from zero and one, imply
\[
    \|\tilde{s}_t(\theta_0)-s_t^*(\theta_0)\|\leq C\rho_1^t.
\]
Consequently,
\[
    \frac1{\sqrt T}\left\|\sum_{t=1}^T\bigl(\tilde{s}_t(\theta_0)-s_t^*(\theta_0)\bigr)\right\|\to0,
\]
and hence
\begin{equation}\label{eq_score_clt_initialized}
    \frac1{\sqrt T}\frac{\partial\mathcal Q_T(\theta_0)}{\partial\theta}
    \xrightarrow[T\to\infty]{d} \mathcal N(0,H_0).
\end{equation}

\smallskip

To analyze the Hessian, let
\begin{equation*}
\begin{split}
\mathcal H_t^*(\theta)=\frac{\partial^2\ell_t^*(\theta)}{\partial\theta\partial\theta^{\T}}
&=\sum_{i=1}^{N} \frac{y_{i,t}-p_{i,t}^*(\theta)}{p_{i,t}^*(\theta)(1-p_{i,t}^*(\theta))}\frac{\partial^2}{\partial\theta\partial\theta^{\T}}p_{i,t}^*(\theta)\\
     &-\sum_{i=1}^{N}\left(\frac{y_{i,t}}{p_{i,t}^*(\theta)^2}+\frac{1-y_{i,t}}{(1-p_{i,t}^*(\theta))^2}\right)
     \left(\frac{\partial}{\partial\theta}p_{i,t}^*(\theta)\right)\left(\frac{\partial}{\partial\theta}p_{i,t}^*(\theta)\right)^{\T} .
\end{split}
\end{equation*}
We need to show the uniform convergence of the Hessian to its mean, $\mathbb{E}\mathcal{H}_t^*(\theta)$. Note that
\begin{equation}\label{eq_expected_hessian}
    \mathbb{E}\mathcal H_t^*(\theta_0)=-H_0.
\end{equation}
As in the proof of the uniform convergence of the criterion in Theorem \ref{th_consistency}, truncate the probability recursion together with its first and second derivative
recursions $M$ periods in the past, and let $\mathcal H_{t,M}^*(\theta)$ denote the resulting Hessian. Repeating the same steps: LLN for $\mathcal H_{t,M}^*(\theta)$ and forgetting of the initial value of $\mathcal H_{t,M}^*(\theta)$, we obtain
\begin{equation}\label{eq_hessian_ulln}
    \sup_{\theta\in\Theta_0} \left\| \frac1T\sum_{t=1}^T\mathcal H_t^*(\theta)-\mathbb{E}\mathcal H_t^*(\theta)\right\|
    \xrightarrow[T\to\infty]{p}0.
\end{equation}
Let $\tilde{\mathcal H}_t(\theta)$ be the Hessian formed from the initialized filter. By probability, first-derivative, and second-derivative forgetting past,
\[
    \sup_{\theta\in\Theta_0}\|\tilde{\mathcal H}_t(\theta)-\mathcal H_t^*(\theta)\|\xrightarrow[t\to\infty]{a.s.}0.
\]
Therefore,
\[
    \sup_{\theta\in\Theta_0} \left\|\frac1T\sum_{t=1}^T\tilde{\mathcal H}_t(\theta)-\frac1T\sum_{t=1}^T\mathcal H_t^*(\theta)\right\|
    \leq\frac1{T}\sum_{t=1}^T \sup_{\theta\in\Theta_0}\|\tilde{\mathcal H}_t(\theta)-\mathcal H_t^*(\theta)\|
    \xrightarrow[T\to\infty]{a.s}0.
\]
Combining the preceding with \eqref{eq_hessian_ulln},
\begin{equation}\label{eq_initialized_hessian_ulln}
    \sup_{\theta\in\Theta_0}  \left\|\frac1{T}\frac{\partial^2\mathcal Q_T(\theta)}{\partial\theta\partial\theta^{\T}}
       -\mathbb{E}\mathcal{H}_t^*(\theta) \right\|\xrightarrow[T\to\infty]{p}0.
\end{equation}

\smallskip

Finally, we can combine the score and the Hessian asymptotics to obtain the asymptotic normality of $\hat{\theta}^{MLE}$. By Theorem \ref{th_consistency}, $\hat\theta^{MLE}\xrightarrow[T\to\infty]{p}\theta_0$. Since $\theta_0$ is an interior point of $\Theta$, with probability approaching one $\hat\theta^{MLE}\in\Theta_0$ and satisfies
\[
    \frac{\partial\mathcal Q_T(\hat\theta^{MLE})}{\partial\theta}=0.
\]
By the mean value theorem,
$$0=\frac1{\sqrt{T}}\frac{\partial\mathcal Q_T(\theta_0)}{\partial\theta}+\frac1{T}\frac{\partial^2\mathcal Q_T(\tilde\theta)}{\partial\theta\partial\theta^{\T}}\sqrt{T}(\hat\theta^{MLE}-\theta_0),$$
where each component of $\frac{\partial\mathcal Q_T(\theta_0)}{\partial\theta}$ corresponds to a separate value of $\tilde\theta$ on a segment connecting $\theta_0$ and $\hat\theta^{MLE}$.
Combining the score CLT \eqref{eq_score_clt_initialized}, the Hessian limit \eqref{eq_expected_hessian} and uniform LLN \eqref{eq_initialized_hessian_ulln}, and consistency of $\hat\theta^{MLE}$, we get
$$\sqrt{T}(\hat\theta^{MLE}-\theta_0)
=-\left(\frac1{T}\frac{\partial^2\mathcal Q_T(\tilde\theta)}{\partial\theta\partial\theta^{\T}}\right)^{-1}\frac1{\sqrt{T}}\frac{\partial\mathcal Q_T(\theta_0)}{\partial\theta}
\xrightarrow[T\to\infty]{d}H_0^{-1}\mathcal{N}(0,H_0)=\mathcal{N}(0,H_0^{-1}).\qedhere$$
\end{proof}

\section{Data}\label{append_data}
As of December 27, 2025, the S$\&$P $100$ consists of the following stock tickers: AAPL, ABBV, ABT, ACN,	ADBE, AIG, AMD, AMGN, AMT, AMZN, AVGO, AXP, BA, BAC, BK, BKNG, BLK, BMY, BRK$\underline{ }$B, C, CAT, CL, CMCSA, COF, COP, COST, CRM, CSCO, CVS, CVX, DE, DHR, DIS, DUK, EMR, FDX, GD, GE, GM, GILD, GOOGL, GS, HD, HON, IBM, INTC, INTU, ISRG, JNJ, JPM, KO, LIN, LLY, LMT, LOW, MA, MCD, MDLZ, MDT, MET, META, MMM, MO, MRK, MS, MSFT, NEE, NFLX, NKE, NOW, NVDA, ORCL, PEP, PFE, PG, PM, PLTR, PYPL, QCOM, RTX, SBUX, SCHW, SO, SPG, T, TGT, TMO, TMUS, TSLA, TXN, UBER, UNH, UNP, UPS, USB, V, VZ, WFC, WMT, XOM.

Six of the above companies are not available for the entire period under consideration, 01.01.2012--12.31.2025. The tickers of those companies are: ABBV, META, NOW, PLTR, PYPL, and UBER.

\bibliographystyle{ecca}
\bibliography{BGARCH_biblio}

\end{document}